\documentclass[reqno,10pt,a4paper,dvips]{amsart}

\usepackage{amssymb,mathptmx,cite,psfrag,eucal,array,setspace,color,geometry,enumitem}
\usepackage{rotating}
\usepackage{graphicx}
\usepackage{tikz}

\DeclareSymbolFont{largesymbols}{OMX}{zplm}{m}{n} %Replaces summation symbol in times by the palatino one...

\numberwithin{equation}{section}

\newcolumntype{C}{>{$}c<{$}} %Defines math mode in tabular (array package)...

\allowdisplaybreaks

\newcommand{\eps}{\varepsilon}

\newcommand{\alg}[1]{\mathfrak{#1}}
\newcommand{\grp}[1]{\mathsf{#1}}

\newcommand{\func}[2]{#1 \left( #2 \right)}
\newcommand{\tfunc}[2]{#1 \bigl( #2 \bigr)}

\newcommand{\brac}[1]{\left( #1 \right)}
\newcommand{\tbrac}[1]{\bigl( #1 \bigr)}
\newcommand{\sqbrac}[1]{\left[ #1 \right]}
\newcommand{\set}[1]{\left\{ #1 \right\}}

\newcommand{\abs}[1]{\left| #1 \right|}

\newcommand{\ZZ}{\mathbb{Z}}
\newcommand{\NN}{\mathbb{N}}
\newcommand{\RR}{\mathbb{R}}

\newcommand{\dd}{\mathrm{d}}
\newcommand{\ii}{\mathfrak{i}}
\newcommand{\ee}{\mathsf{e}}

\newcommand{\wun}{\mathbf{1}}

\newcommand{\killing}[2]{\kappa \bigl( #1 , #2 \bigr)}

\newcommand{\acomm}[2]{\bigl\{ #1 , #2 \bigr\}}
\newcommand{\comm}[2]{\bigl[ #1 , #2 \bigr]}

\newcommand{\ket}[1]{\bigl\lvert #1 \bigr\rangle}

\newcommand{\affine}[1]{\widehat{#1}}

\newcommand{\DiscMod}[1]{\mathcal{D}_{#1}}
\newcommand{\IrrMod}[1]{\mathcal{L}_{#1}}
\newcommand{\StagMod}[1]{\mathcal{S}_{#1}}
\newcommand{\TypMod}[1]{\mathcal{E}_{#1}}
\newcommand{\VerMod}[1]{\mathcal{V}_{#1}}

\newcommand{\conjaut}{\mathsf{w}} % The conjugation automorphism
\newcommand{\sfaut}{\sigma} % The spectral flow automorphism
\newcommand{\conjmod}[1]{\tfunc{\conjaut}{#1}} % Conjugate a module
\newcommand{\sfmod}[2]{\tfunc{\sfaut^{#1}}{#2}} % Apply spectral flow #1 times to #2

\newcommand{\SLA}[2]{\alg{#1} \left( #2 \right)}
\newcommand{\SLSA}[3]{\alg{#1} \left( #2 \middle\vert #3 \right)}
\newcommand{\AKMA}[2]{\affine{\alg{#1}} \left( #2 \right)}
\newcommand{\AKMSA}[3]{\affine{\alg{#1}} \left( #2 \middle\vert #3 \right)}
\newcommand{\SLG}[2]{\grp{#1} \left( #2 \right)}

\newcommand{\minmod}[2]{\mathsf{M} \left( #1 , #2 \right)}

\newcommand{\traceover}[1]{\tr_{\raisebox{-3pt}{$\scriptstyle #1$}}}

\newcommand{\Gr}[1]{\Bigl[ #1 \Bigr]} % element of a Grothendieck ring
\newcommand{\tGr}[1]{\bigl[ #1 \bigr]}
\newcommand{\KSGr}[1]{\Bigl\langle #1 \Bigr\rangle} % Koh-Sorba Grothendieck ring element
\newcommand{\tKSGr}[1]{\bigl\langle #1 \bigr\rangle}

\newcommand{\chmap}{\mathrm{ch}}
\newcommand{\ch}[1]{\chmap \tGr{#1}}
\newcommand{\fch}[2]{\ch{#1} \bigl( #2 \bigr)}
\newcommand{\vch}[1]{\chi^{\mathrm{Vir}}_{#1}}
\newcommand{\fvch}[2]{\vch{#1} \bigl( #2 \bigr)}

\newcommand{\modS}{\mathsf{S}}
\newcommand{\modT}{\mathsf{T}}
\newcommand{\vmodS}{\mathsf{S}^{\mathrm{Vir}}}

\newcommand{\modarg}[3]{\left( \: #1 \: \middle\vert \: #2 \: \middle\vert \: #3 \: \right)}

\newcommand{\fuse}{\mathbin{\times}}

\newcommand{\fuscoeff}[2]{\mathsf{N}_{#1}^{\hphantom{#1} #2}}
\newdimen{\Virwidth}
\newcommand{\vfuscoeff}[2]{
  \settowidth{\Virwidth}{$\mathrm{Vir}$}
  \mathsf{N}_{#1}^{\mathrm{Vir} \hspace{-\Virwidth} \hphantom{#1} #2}
}
\newcommand{\vpfuscoeff}[3]{
  \settowidth{\Virwidth}{$v$}
  \mathsf{N}_{#1}^{#3 \hspace{-\Virwidth} \hphantom{#1} #2}
}

\newcommand{\fusring}[1]{\mathcal{F}_{#1}}

\newcommand{\normord}[1]{\mbox{${} : #1 : {}$}} % {} necessary to prevent := or =:

\newcommand{\jth}[1]{\vartheta_{#1}}
\newcommand{\Jth}[2]{\jth{#1} \bigl( #2 \bigr)}

\newcommand{\lra}{\longrightarrow}
\newcommand{\dses}[3]{0 \lra #1 \lra #2 \lra #3 \lra 0}

\newcommand{\eqnref}[1]{Equation~\eqref{#1}}
\newcommand{\eqnDref}[2]{Equations~\eqref{#1} and \eqref{#2}}
\newcommand{\secref}[1]{Section~\ref{#1}}

\newcommand{\figref}[1]{Figure~\ref{#1}}
\newcommand{\tabref}[1]{Table~\ref{#1}}
\newcommand{\thmref}[1]{Theorem~\ref{#1}}
\newcommand{\propref}[1]{Proposition~\ref{#1}}
\newcommand{\propQref}[4]{Propositions~\ref{#1}, \ref{#2}, \ref{#3} and \ref{#4}}
\newcommand{\lemref}[1]{Lemma~\ref{#1}}
\newcommand{\corref}[1]{Corollary~\ref{#1}}

\newcommand{\cft}{conformal field theory}
\newcommand{\cfts}{conformal field theories}
\newcommand{\uea}{universal enveloping algebra}

\newcommand{\lcfts}{logarithmic conformal field theories}
\newcommand{\WZW}{Wess-Zumino-Witten}

\newcommand{\opes}{operator product expansions}
\newcommand{\hws}{highest weight state}
\newcommand{\hwss}{highest weight states}

\newcommand{\hwm}{highest weight module}
\newcommand{\hwms}{highest weight modules}

\DeclareMathOperator{\tr}{tr}
\DeclareMathOperator{\id}{id}

\theoremstyle{plain}
\newtheorem{thm}{Theorem}
\newtheorem{prop}[thm]{Proposition}
\newtheorem{lem}[thm]{Lemma}
\newtheorem{cor}[thm]{Corollary}
\newtheorem*{ex}{Example}
\newtheorem*{conj}{Conjecture}

\begin{document}

\title[Modular Data and Verlinde Formulae for Fractional Level WZW Models II]{Modular Data and Verlinde Formulae \\ for Fractional Level WZW Models II}

\author[T Creutzig]{Thomas Creutzig}

\address[T Creutzig]{
Department of Mathematical and Statistical Sciences\\
University of Alberta \\
Edmonton, Alberta  T6G 2G1 \\
Canada
}

\email{creutzig@ualberta.ca}

\author[D Ridout]{David Ridout}

\address[David Ridout]{
Department of Theoretical Physics \\
Research School of Physics and Engineering;
and
Mathematical Sciences Institute;
Australian National University \\
Canberra, ACT 0200 \\
Australia
}

\email{david.ridout@anu.edu.au}

\thanks{\today}

\begin{abstract}
This article gives a complete account of the modular properties and Verlinde formula for conformal field theories based on the affine Kac-Moody algebra $\AKMA{sl}{2}$ at an arbitrary admissible level $k$.  Starting from spectral flow and the structure theory of relaxed highest weight modules, characters are computed and modular transformations are derived for every irreducible admissible module.  The culmination is the application of a continuous version of the Verlinde formula to deduce non-negative integer structure coefficients which are identified with Grothendieck fusion coefficients.  The Grothendieck fusion rules are determined explicitly.  These rules reproduce the well-known ``fusion rules'' of Koh and Sorba, negative coefficients included, upon quotienting the Grothendieck fusion ring by a certain ideal.
\end{abstract}

\maketitle

\onehalfspacing

\section{Introduction} \label{sec:Intro}

This is the sequel to the article \cite{CreMod12} devoted to solving the longstanding problem of determining the (Grothendieck) fusion coefficients, for admissible level $\AKMA{sl}{2}$ \WZW{} models, from a formula of Verlinde type.  The main issue here is that initial attempts to do so, using the standard Verlinde formula for \hwms{} \cite{VerFus88}, led to certain ``fusion coefficients'' being \emph{negative} integers \cite{KohFus88} (we refer to \cite{CreMod12} for further historical detail).  The mechanism responsible for these negative coefficients was only obtained recently \cite{RidSL208} for the admissible level $k=-\tfrac{1}{2}$.  There, it was pointed out that this negativity resulted from assuming that the irreducible modules of the spectrum were all highest weight and from not properly accounting for the regions of convergence of the \hwms{}' characters (see \cite{RidSL210,RidFus10} for a more detailed discussion).

While this mechanism accounts for what goes wrong in applying the standard Verlinde formula, the problem of how to modify this formula so as to obtain non-negative integer fusion coefficients remained.  This was addressed in \cite{CreMod12} wherein the modular properties of the $\AKMA{sl}{2}$ models at levels $k=-\tfrac{1}{2}$ and $k=-\tfrac{4}{3}$ were analysed.  The main result was that a continuous version of the Verlinde formula may be applied to each of these theories and that the results were consistent with the known fusion rules (which have only been computed for these levels \cite{GabFus01,RidFus10}).  In particular, the continuum Verlinde formula yielded non-negative integers that precisely reproduced the Grothendieck fusion coefficients.  The aim of this article is to generalise the continuum Verlinde computations to all admissible levels, for $\AKMA{sl}{2}$ at least, and show that the mechanism identified to generate the negative ``fusion coefficients'' when $k=-\tfrac{1}{2}$ is also responsible in this greater generality.

The methodology employed here to tame the modular properties of fractional level \WZW{} models is but one instance of a general programme we are developing (see \cite{CreLog13} for a review) to deal with Verlinde formulae for \emph{logarithmic} \cfts{}.  Indeed, it is known that the $\AKMA{sl}{2}$ models with $k=-\tfrac{1}{2}$ and $k=-\tfrac{4}{3}$ are necessarily logarithmic \cite{GabFus01,LesLog04,RidSL210} and this is surely the case more generally.  This programme is, in some respects, a far-reaching extension to general \lcfts{} of ideas which were originally developed in the string theory literature to deal with supersymmetric and non-compact spacetimes (see \cite{RozSTM93,MalStr01,SalGL106,QueFre07} for example).  Besides the $\AKMA{sl}{2}$ theories considered here, this programme has already been successfully applied to the Grothendieck fusion rules of $\AKMSA{gl}{1}{1}$ \cite{CreRel11}, its extended algebras \cite{Alfes:2012pa} and its Takiff version \cite{BabTak12}, the $\brac{1,p}$ singlet and triplet models \cite{CreLog13,CreMpq13} and even the Virasoro algebra \cite{MorVir13}.

We begin, as always, with notation and conventions.  \secref{sec:SL2Reps} describes this for $\AKMA{sl}{2}$ and its \hwms{} before introducing the conjugation and spectral flow automorphisms which play such a vital role in what follows.  \secref{sec:Admissible} defines the notion of admissibility, first for the level $k$ and then for $\AKMA{sl}{2}_k$-modules.  Theorems of Adamovi\'{c} and Milas are then quoted \cite{AdaVer95} giving the irreducible admissibles in the category of \hwms{} and the category of \emph{relaxed} \hwms{}.  We then introduce an analogue of the Kac table familiar from the Virasoro minimal models to organise the admissible irreducibles.  Finally, we extend our collection of admissibles using spectral flow and catalogue the relationships between spectral flow versions of irreducible admissibles.  At this point, we define appropriate notions (following \cite{CreLog13}) of ``standard'', ``typical'' and ``atypical'' modules.  In this setting, all highest weight admissibles are atypical and a standard module is typical if and only if it is irreducible.

Our first main result is the character formula for a general standard module.  Unlike the characters of the \hwms{}, the standard characters do not converge anywhere and must be represented as distributions.  The result, given in \secref{sec:TypChar} (\propref{prop:ChErs} and \corref{cor:ChTyp}), describes the character as a sum of delta functions weighted by Virasoro minimal model characters.  This is surely a manifestation of quantum hamiltonian reduction \cite{FeiAff92,deBRel94} and it lifts the observation of \cite{MukFra90}, where it was noticed that residues of admissible highest weight characters involved minimal model characters, to a much more elegant setting.  The modular transformation rules of the standard characters are then computed in \secref{sec:ModTyp} (\thmref{thm:TypMod}) and we verify that one obtains a (projective) representation of the modular group of uncountably-infinite dimension.  Moreover, the ``S-matrix'' is seen to be symmetric and unitary.

\secref{sec:AtypChar} then addresses the atypical characters.  We wish to determine them as distributions so as to avoid the convergence issues that stymied progress for so long, so we derive resolutions for each atypical module in terms of reducible but indecomposable (atypical) standard modules.  The resulting character formulae then allow us to compute the modular transformation rules of (certain) atypical characters in \secref{sec:ModAtyp} (\thmref{thm:AtypMod}).  In particular, we obtain the S-transformation of the vacuum character (the vacuum module is highest weight, hence atypical).  These atypical computations rely on a rather ungainly identity (\lemref{lem:UsefulIdentity}) whose representation-theoretic significance is not yet apparent to us.  Presumably, generalising these results to higher rank affine Kac-Moody algebras will clear this up.

In any case, we now have all the ingredients to apply the obvious continuum analogue of the Verlinde formula.  Assuming that this does yield the Grothendieck fusion coefficients, we then compute the complete set of Grothendieck fusion rules explicitly.  This is detailed in \secref{sec:Verlinde} (see \propQref{prop:FusTyp}{prop:FusIrrTyp}{prop:FusIrrIrr}{prop:FusDrs}).  When we can be sure that the corresponding fusion products are completely reducible, these results can be immediately lifted to the fusion ring itself.  In this way, we prove (\thmref{thm:FusionSubring}) that the fusion ring of an admissible level theory always contains a subring isomorphic to that of a particular non-negative integer level theory.  One consequence is that one obtains, for almost all admissible levels, a non-trivial simple current generalising that which gives the $\beta \gamma$ ghosts in the $k=-\tfrac{1}{2}$ theory \cite{RidSL208}.

Another consequence of our explicit computations is that all the Grothendieck fusion coefficients, as computed by the continuum Verlinde formula, are \emph{non-negative integers} (\thmref{thm:GrFusCoeffPos}).  Because the resolutions we have used lead to alternating sums for atypical characters in terms of standard ones, this non-negativity result is highly non-trivial and represents a very strong endorsement of our claim that the continuum Verlinde formula does indeed give the Grothendieck fusion coefficients correctly.  A second strong endorsement is discussed in \secref{sec:KohSorba} where we recover the ``fusion rules'' of \cite{KohFus88}, negative coefficients and all, for all admissible levels $k$, by applying the mechanism explained in \cite{RidSL208} to our Grothendieck fusion rules.  These two endorsements give us complete confidence that we have solved the longstanding problem of modular properties and Verlinde formulae for fractional level \WZW{} models.

Throughout the text, we illustrate our results by applying them to the levels $k=-\tfrac{1}{2}$ and $k=-\tfrac{4}{3}$, thereby checking against what was reported in \cite{CreMod12}.  \secref{sec:Examples} concludes the article by discussing three other admissible level theories which are also of independent interest.  In each case, we exhaustively describe the Grothendieck fusion rules and compute the extended algebra defined by the simple current guaranteed by \thmref{thm:FusionSubring}.  When $k=-\tfrac{5}{4}$, we obtain in this way a conformal embedding of $\AKMA{sl}{2}_{-5/4}$ into $\AKMSA{osp}{1}{2}_{-5/4}$.  When $k=-\tfrac{2}{3}$, the extended algebra is the reduced $N=3$ superconformal algebra at $c=-\tfrac{3}{2}$.  Finally, $k=\tfrac{1}{2}$ yields an interesting simple current extension that we tentatively identify with the quantum hamiltonian reduction of $\affine{\alg{g}}_{2,-3/2}$.

Of course, there are many points that remain to be addressed.  First, it is clear that one should be able to generalise our results to higher rank fractional level affine Kac-Moody algebras and superalgebras and it would be extremely interesting to do so.  Moreover, the relationship (if any) between these fractional level models and the \WZW{} models on non-compact Lie groups requires clarification.  Even at the level of $\AKMA{sl}{2}$, there are many fascinating questions still to consider, for example, that of classifying modular invariant partition functions for the admissible level theories.  Mathematically, one should also ask after homological characterisations of the spectrum:  What is the physical category of modules?  Which modules are projective in this category?  Which are rigid?  Can we characterise admissible staggered modules as was done for the Virasoro algebra in \cite{RidSta09}?  Even more interesting, and perhaps more relevant for comparison with non-compact target space models, what happens if we relax the irreducibility of the vacuum module?  It is clear that the study of logarithmic theories with affine symmetries will remain rich and rewarding.  We hope to report further on this study in the future.

\section{$\AKMA{sl}{2}$ and its Representations} \label{sec:SL2Reps}

Consider the simple complex Lie algebra $\SLA{sl}{2}$ and its standard basis elements
\begin{equation}
E = 
\begin{pmatrix}
0 & 1 \\
0 & 0
\end{pmatrix}
, \qquad H = 
\begin{pmatrix}
1 & 0 \\
0 & -1
\end{pmatrix}
, \qquad F = 
\begin{pmatrix}
0 & 0 \\
1 & 0
\end{pmatrix}
.
\end{equation}
This basis is tailored to a triangular decomposition respecting the adjoint (conjugate transpose) that picks out the real form $\SLA{su}{2}$.  Indeed, the Cartan element $H$ is clearly self-adjoint and the raising and lowering operators $E$ and $F$ are swapped by the adjoint.  In what follows, we want to study conformal field theories whose symmetry algebras are the affine Kac-Moody algebras $\AKMA{sl}{2}$ at levels $k$ which are not non-negative integers.  The well-known quantisation of the level for the Wess-Zumino-Witten model on $\SLG{SU}{2}$ suggests that one should not lift the $\SLA{su}{2}$ adjoint to $\AKMA{sl}{2}$.  Instead, the absence of level-quantisation for $\SLG{SL}{2;\RR}$ leads us to propose lifting the adjoint that picks out the other real form $\SLA{sl}{2;\RR}$.

The $\SLA{sl}{2;\RR}$ adjoint simply negates the basis elements $E$, $H$ and $F$, hence may be described as negation followed by complex conjugation:  $J^{\dag} = -J^*$.  This means that this basis is not suited to triangular decompositions that respect the $\SLA{sl}{2;\RR}$ adjoint.  For this reason, we choose a new basis $\set{e,h,f}$ of $\SLA{sl}{2}$:
\begin{equation}
e = \frac{1}{2} 
\begin{pmatrix}
-1 & \ii \\
\ii & 1
\end{pmatrix}
, \qquad h = 
\begin{pmatrix}
0 & \ii \\
-\ii & 0
\end{pmatrix}
, \qquad f = \frac{1}{2} 
\begin{pmatrix}
1 & \ii \\
\ii & -1
\end{pmatrix}
.
\end{equation}
Because $e^{\dag} = f$ and $h^{\dag} = h$ with respect to the $\SLA{sl}{2;\RR}$ adjoint, this basis is suited to the desired triangular decomposition.  Note that the non-vanishing commutation relations in this basis are
\begin{equation}
\comm{h}{e} = 2 e, \qquad \comm{e}{f} = -h, \qquad \comm{h}{f} = -2 f.
\end{equation}
Similarly, the trace form in this basis attracts an unfamiliar sign:
\begin{equation}
\killing{h}{h} = 2, \qquad \killing{e}{f} = \killing{f}{e} = -1.
\end{equation}
We remark that choosing the adjoint correctly is not just mathematical sophistry --- this choice plays a subtle, but vital, role in many aspects of the representation theory, unitarity being the most obvious.  An example of this subtlety appears in the $k=-\tfrac{1}{2}$ theory which has a simple current extension which fails to be associative when the $\SLA{su}{2}$ adjoint is chosen \cite{RidSL208}.  The associative extension one obtains with the $\SLA{sl}{2;\RR}$ adjoint is, of course, the $\beta \gamma$ ghost system (see \secref{sec:Examples}).

The commutation relations of the affine Kac-Moody algebra $\AKMA{sl}{2}$ are therefore
\begin{equation}
\begin{aligned}
\comm{h_m}{e_n} &= +2 e_{m+n}, \\
\comm{h_m}{f_n} &= -2 f_{m+n},
\end{aligned}
\qquad
\begin{aligned}
\comm{h_m}{h_n} &= 2m \delta_{m+n,0} K, \\
\comm{e_m}{f_n} &= -h_{m+n} - m \delta_{m+n,0} K,
\end{aligned}
\qquad
\begin{aligned}
\comm{e_m}{e_n} &= 0, \\
\comm{f_m}{f_n} &= 0.
\end{aligned}
\end{equation}
where $K$ is central.  We will habitually replace $K$ by its common eigenvalue $k$, the level, when acting upon the modules comprising each theory.\footnote{Technically, we should do this in the \uea{} by quotienting by the ideal generated by $K - k \wun$.  Doing this at the level of the Lie algebra is a standard sloppiness which leads to no harm.}  With this replacement, the Sugawara construction gives the standard energy-momentum tensor
\begin{equation} \label{eqn:DefT}
\func{T}{z} = \frac{1}{2 \brac{k+2}} \brac{\frac{1}{2} \normord{\func{h}{z} \func{h}{z}} - \normord{\func{e}{z} \func{f}{z}} - \normord{\func{f}{z} \func{e}{z}}},
\end{equation}
at least when $k \neq -2$.  The modes $L_n$ of $\func{T}{z}$ then generate a copy of the Virasoro algebra of central charge
\begin{equation}
c = \frac{3k}{k+2} = 3 - \frac{6}{t}.
\end{equation}
Here, we take the opportunity to introduce the notation $t = k+2$.

The triangular decomposition that we have chosen for $\SLA{sl}{2}$ lifts, in the standard manner, to one for $\AKMA{sl}{2}$.  The notions of \hwss{} and Verma modules are then available.  An easy consequence of \eqref{eqn:DefT} is that a \hws{} of weight ($h_0$-eigenvalue) $\lambda$ will have conformal dimension ($L_0$-eigenvalue)
\begin{equation}
\Delta_{\lambda} = \frac{\lambda \brac{\lambda + 2}}{4 \brac{k+2}} = \frac{\brac{\lambda + 1}^2 - 1}{4t}.
\end{equation}
We will denote the Verma module generated by a \hws{} of weight $\lambda$ by $\VerMod{\lambda}$.  The irreducible quotient of $\VerMod{\lambda}$ will be denoted by $\IrrMod{\lambda}$ if $\lambda \in \NN$, and by $\DiscMod{\lambda}^+$ otherwise.  The notation here is chosen to reflect the nature of the zero-grade subspace (the states of minimal conformal dimension) of the irreducible as an $\SLA{sl}{2}$-module.  When $\lambda \in \NN$, this subspace forms a finite-dimensional irreducible $\SLA{sl}{2}$-module, whereas it forms an infinite-dimensional irreducible of the discrete series type otherwise.  We will refer to $\IrrMod{0}$ as the vacuum module and its \hws{} $\ket{0}$ as the vacuum in what follows.

The subgroup of automorphisms of $\AKMA{sl}{2}$ which leave the span of the zero-modes $h_0$, $K$ and $L_0$ invariant is isomorphic to $\ZZ_2 \ltimes \ZZ$.  We take the order two generator to be the \emph{conjugation} automorphism $\conjaut$ which is the Weyl reflection corresponding to the finite simple root.  The infinite order generator is the \emph{spectral flow} automorphism $\sfaut$ which may be regarded as a square root of the affine Weyl translation by the (finite) simple coroot (in fact, $\sfaut$ is translation by the dual of the finite simple root).  These automorphisms fix $K$, hence the level $k$ is preserved, and otherwise act as follows:
\begin{equation}
\begin{aligned}
\func{\conjaut}{e_n} &= f_n, \\ 
\func{\sfaut^{\ell}}{e_n} &= e_{n-\ell},
\end{aligned}
\qquad
\begin{aligned}
\func{\conjaut}{h_n} &= -h_n, \\
\func{\sfaut^{\ell}}{h_n} &= h_n - \delta_{n,0} \ell k,
\end{aligned}
\qquad
\begin{aligned}
\func{\conjaut}{f_n} &= e_n, \\
\func{\sfaut^{\ell}}{f_n} &= f_{n+\ell},
\end{aligned}
\qquad
\begin{aligned}
\func{\conjaut}{L_0} &= L_0, \\
\func{\sfaut^{\ell}}{L_0} &= L_0 - \tfrac{1}{2} \ell h_0 + \tfrac{1}{4} \ell^2 k.
\end{aligned}
\end{equation}
The normality of the subgroup generated by $\sfaut$ follows from $\conjaut \sfaut = \sfaut^{-1} \conjaut$.

One important use for these automorphisms is to modify the action of $\AKMA{sl}{2}$ on any module $\mathcal{M}$, thereby obtaining new modules $\tfunc{\conjaut^*}{\mathcal{M}}$ and $\tfunc{\sfaut^*}{\mathcal{M}}$.  The first is precisely the module conjugate to $\mathcal{M}$ --- its weights are the negatives of the weights of $\mathcal{M}$, though the conformal dimensions remain unchanged.  The second is called the \emph{spectral flow} image of $\mathcal{M}$ --- its weights have been shifted by a fixed amount, but its conformal dimensions also change.  Explicitly, the modified algebra action defining these new modules is given by
\begin{equation}
J \cdot \conjaut^* \ket{v} = \func{\conjaut^*}{\func{\conjaut^{-1}}{J} \ket{v}}, \qquad J \cdot \sfaut^* \ket{v} = \func{\sfaut^*}{\func{\sfaut^{-1}}{J} \ket{v}} \qquad \text{(\(J \in \AKMA{sl}{2}\)).}
\end{equation}
It is easy to check that if $\ket{\lambda , \Delta} \in \mathcal{M}$ is a state of weight $\lambda$ and conformal dimension $\Delta$, then the state $\tbrac{\sfaut^{\ell}}^* \ket{\lambda , \Delta} \in \tfunc{\tbrac{\sfaut^{\ell}}^*}{\mathcal{M}}$ satisfies
\begin{equation}
\begin{split}
h_0 \tbrac{\sfaut^{\ell}}^* \ket{\lambda , \Delta} &= \brac{\lambda + \ell k} \tbrac{\sfaut^{\ell}}^* \ket{\lambda , \Delta}, \\
L_0 \tbrac{\sfaut^{\ell}}^* \ket{\lambda , \Delta} &= \brac{\Delta + \frac{1}{2} \ell \lambda + \frac{1}{4} \ell^2 k} \tbrac{\sfaut^{\ell}}^* \ket{\lambda , \Delta}.
\end{split}
\end{equation}
In what follows, we will usually omit the superscript ``$*$'' which distinguishes the induced spectral flow maps between modules from the spectral flow algebra automorphisms.  Which is meant should be clear from the context.

\section{Admissible Levels and Modules} \label{sec:Admissible}

Recall that when the level $k$ is a non-negative integer, the (chiral) spectrum of any \cft{} with $\AKMA{sl}{2}$ symmetry and an irreducible vacuum module may only contain the irreducible modules $\IrrMod{\lambda}$ with $\lambda = 0, 1, \ldots, k$.  This is the spectrum of the \WZW{} model on $\SLG{SU}{2}$.  The reason boils down to the following fact:  Let $\ket{v_0}$ denote the \hws{} of the vacuum Verma module $\VerMod{0}$.  As $k \in \NN$, $\VerMod{0}$ possesses a non-trivial singular vector $e_{-1}^{k+1} \ket{v_0}$, meaning that it is not descended from the trivial singular vector $f_0 \ket{v_0}$, which has to be set to zero in order to form the irreducible vacuum module $\IrrMod{0}$.  Setting this singular vector to zero is only consistent with the state-field correspondence of \cft{} if the spectrum is restricted as above.  This seems to have been first explained in \cite{FeiAnn92}, though the argument has since been modified and made rigorous within the formalism of vertex algebras by Zhu \cite{ZhuMod96}.

It is natural to ask if there are other levels at which an irreducible vacuum module similarly constrains the spectrum.  To have such constraints, one needs to know when the corresponding vacuum Verma module has a non-trivial singular vector.  This question may be answered using the Kac-Kazhdan formula \cite{KacStr79} for the determinant of the Shapovalov form in each (affine) weight space.  The result is that such a non-trivial singular vector exists precisely when
\begin{equation}
t = k+2 = \frac{u}{v}, \qquad \text{with \(\gcd \set{u,v} = 1\), \(u \in \ZZ_{\geqslant 2}\) and \(v \in \ZZ_{\geqslant 1}\).}
\end{equation}
Moreover, the singular vector will have weight $2 \brac{u-1}$ and conformal dimension $\brac{u-1} v$.  Levels $k$ satisfying the above conditions are called \emph{admissible}.  Equivalently, $k$ is said to be admissible if the universal vertex algebra corresponding to $\AKMA{sl}{2}_k$ is not simple.

Determining the constraints that this singular vector imposes on the spectrum is not quite as easy.  One has a semi-explicit formula for the singular vector due to Malikov, Feigin and Fuchs \cite{MalSin86}.  However, this formula involves rational powers of the affine modes which must be massaged using analytic continuations of the commutation rules in order to arrive at an explicit expression (see \cite{BauFus93,FurSin94,MatPri99} for concrete examples of such massaging).

\begin{ex}[see \cite{GabFus01}]
The level $k=-\tfrac{4}{3}$ has $t=\tfrac{2}{3}$, hence $u=2$ and $v=3$.  This level is therefore admissible.  Kac-Kazhdan tells us that the non-trivial singular vector in the vacuum Verma module has weight $2$ and conformal dimension $3$.  It is given, in the Malikov-Feigin-Fuchs form, by
\begin{equation}
\ket{\chi'} = e_{-1}^{7/3} f_0^{5/3} e_{-1} f_0^{1/3} e_{-1}^{-1/3} \ket{v_0}.
\end{equation}
For deriving constraints, it is in fact more convenient to consider the descendant $\ket{\chi} = f_0 \ket{\chi'}$ whose weight is $0$.  This state will also be set to $0$ in the irreducible vacuum module.  Massaging the above expression appropriately leads to the (renormalised) explicit form
\begin{equation}
\ket{\chi} = \brac{9 h_{-1}^3 + 18 h_{-2} h_{-1} - 16 h_{-3} - 36 f_{-1} h_{-1} e_{-1} -24 e_{-2} f_{-1} + 96 f_{-2} e_{-1}} \ket{v_0}.
\end{equation}
The field $\func{\chi}{z}$, and so its zero-mode $\chi_0$, must therefore act as $0$ on the spectrum.  But, applying $\chi_0$ to a \hws{} $\ket{v_{\lambda}}$ of weight $\lambda$ gives
\begin{equation}
\chi_0 \ket{v_{\lambda}} = \lambda \brac{3 \lambda + 2} \brac{3 \lambda + 4} \ket{v_{\lambda}},
\end{equation}
hence we conclude that the only \hwss{} allowed are those with weights $0$, $-\tfrac{2}{3}$ and $-\tfrac{4}{3}$.  It follows that the only \hwms{} in the spectrum are the irreducibles $\IrrMod{0}$, $\DiscMod{-2/3}^+$ and $\DiscMod{-4/3}^+$.
\end{ex}

\begin{ex}[see \cite{LesLog04,RidSL208}]
For $k=-\tfrac{1}{2}$, we have $u=3$ and $v=2$, so this level is also admissible.  The non-trivial singular vector has weight and conformal dimension $4$, and its zero-weight descendant takes the form
\begin{align}
\ket{\chi} &= f_0^2 e_{-1}^{7/2} f_0^2 e_{-1}^{1/2} \ket{v_0} \notag \\
&= \left( 4 h_{-1}^4 + 4 h_{-2} h_{-1}^2 + 19 h_{-2}^2 - 92 h_{-3} h_{-1} + 9 h_{-4} - 32 f_{-1}^2 e_{-1}^2 - 8 f_{-1} h_{-1}^2 e_{-1} + 100 f_{-2} h_{-1} e_{-1} \right. \notag \\
&\mspace{100mu} \left. + 64 h_{-2} f_{-1} e_{-1} - 68 e_{-2} h_{-1} f_{-1} - 82 f_{-2} e_{-2} - 28 f_{-3} e_{-1} - 124 e_{-3} f_{-1} \right) \ket{v_0}.
\end{align}
The zero-mode of the field $\func{\chi}{z}$ then acts on a \hws{} as
\begin{equation}
\chi_0 \ket{v_{\lambda}} = \lambda \brac{\lambda - 1} \brac{2 \lambda + 1} \brac{2 \lambda + 3} \ket{v_{\lambda}},
\end{equation}
so the allowed \hwms{} are the irreducibles $\IrrMod{0}$, $\IrrMod{1}$, $\DiscMod{-1/2}^+$ and $\DiscMod{-3/2}^+$.
\end{ex}

As these examples show, unpacking the Malikov-Feigin-Fuchs formula for the non-trivial vacuum singular vector is extremely cumbersome.  It is therefore rather remarkable that the constraints upon the spectrum have been worked out for arbitrary admissible levels.  This result is due to Adamovi\'{c} and Milas \cite{AdaVer95} who determined Zhu's algebra using an explicit formula of Fuchs \cite{FucTwo89} for a projection of the non-trivial singular vector onto the \uea{} of $\SLA{sl}{2}$.  Modules which are allowed in the spectrum of an admissible level theory are also said to be \emph{admissible}.\footnote{The original definition of admissibility is that of Kac and Wakimoto \cite{KacMod88} who defined admissible weights in order to derive a generalisation of the Weyl-Kac character formula for integrable modules.  Their admissible weights are precisely the \hwms{} of the admissible modules, as they have been defined here.  We just prefer to arrive at the definition from consideration of the vertex algebra.}  The spectrum of admissible \hwms{} is as follows:

\begin{thm}[Adamovi\'{c}--Milas]
Let $k = t-2$ be an admissible level and let
\begin{equation} \label{eqn:DefLambda}
\lambda_{r,s} = r-1-ts.
\end{equation}
The admissible \hwms{} are then exhausted by the following irreducibles:
\begin{itemize}
\item $\IrrMod{r,0} \equiv \IrrMod{\lambda_{r,0}}$, for $r=1,2,\ldots,u-1$,
\item $\DiscMod{r,s}^+ \equiv \DiscMod{\lambda_{r,s}}^+$, for $r=1,2,\ldots,u-1$ and $s=1,2,\ldots,v-1$.
\end{itemize}
\end{thm}

\noindent Mathematically, admissibility just means that the \hwm{} is a module for the (simple) vertex algebra associated with $\AKMA{sl}{2}$ at the admissible level $k$.  It is convenient to extend this definition of admissibility beyond the highest weight category --- from now on, any vertex algebra module will be termed admissible.  Note that when $v=1$, so $k \in \NN$, the set of $\DiscMod{}^+$-type modules is empty and the admissible \hwms{} are precisely the $\IrrMod{r-1}$ with $r=1,2,\ldots,k+1$.

It is convenient to collect the admissible highest weights $\lambda_{r,s}$ into a table, analogous to the Kac table which gives the allowed conformal dimensions for the \hwss{} of a Virasoro minimal model.  We present some of these tables, both for admissible highest weights $\lambda_{r,s}$ and their conformal dimensions
\begin{equation} \label{eqn:DefConfDimRS}
\Delta_{r,s} = \frac{\brac{r-ts}^2 - 1}{4t} = \frac{\brac{vr-us}^2 - v^2}{4uv},
\end{equation}
in \figref{fig:Tables}.  We note that, if one ignores the left-most column ($s=0$) which describes the $\IrrMod{}$-type admissibles, then these tables have symmetries similar to Kac tables.  In particular, we have
\begin{equation} \label{eqn:KacSymmetry}
\lambda_{u-r,v-s} = -\lambda_{r,s} - 2, \quad \Delta_{u-r,v-s} = \Delta_{r,s} \qquad \text{(\(s \neq 0\)).}
\end{equation}
This similarity between the table of $\DiscMod{}^+$-type admissibles and the Kac table for the minimal model $\minmod{u}{v}$ is more than just analogy.  In particular, note that if we take $t=k+2$ to define a Virasoro central charge and Virasoro conformal dimensions by
\begin{equation}
c^{\mathrm{Vir}} = 13 - 6 \brac{t + t^{-1}}, \qquad \Delta_{r,s}^{\mathrm{Vir}} = \frac{\brac{r-ts}^2 - \brac{1-t}^2}{4t},
\end{equation}
then one finds that
\begin{equation} \label{eqn:ModularMagic}
\Delta_{r,s} - \frac{c}{24} + \frac{1}{12} = \Delta_{r,s}^{\mathrm{Vir}} - c^{\mathrm{Vir}}.
\end{equation}
This relation is the key upon which a large proportion of the following analysis rests.

{
\renewcommand{\arraystretch}{1.5} % Why are latex tables so poorly spaced?
\begin{figure}
\begin{center}
\begin{tikzpicture}
\node at (0,0) [] {
\begin{tabular}{|C|C|}
\hline
0 & -\tfrac{3}{2} \\
1 & -\tfrac{1}{2} \\
\hline
\end{tabular}
};
\node at (2,0) [] {
\begin{tabular}{|C|C|}
\hline
0 & -\tfrac{1}{8} \\
\tfrac{1}{2} & -\tfrac{1}{8} \\
\hline
\end{tabular}
};
\node at (1,-1.2) [] {$k=-\tfrac{1}{2}$};
\node at (0,1) [] {$\lambda_{r,s}$};
\node at (2,1) [] {$\Delta_{r,s}$};
\node at (0,-3.7) [] {
\begin{tabular}{|C|C|}
\hline
0 & -\tfrac{5}{2} \\
1 & -\tfrac{3}{2} \\
2 & -\tfrac{1}{2} \\
3 & \tfrac{1}{2} \\
\hline
\end{tabular}
};
\node at (2,-3.7) [] {
\begin{tabular}{|C|C|}
\hline
0 & \tfrac{1}{8} \\
\tfrac{3}{10} & -\tfrac{3}{40} \\
\tfrac{4}{5} & -\tfrac{3}{40} \\
\tfrac{3}{2} & \tfrac{1}{8} \\
\hline
\end{tabular}
};
\node at (1,-5.5) [] {$k=\tfrac{1}{2}$};
\node at (0,-2) [] {$\lambda_{r,s}$};
\node at (2,-2) [] {$\Delta_{r,s}$};
\node at (5.5,0) [] {
\begin{tabular}{|C|CC|}
\hline
0 & -\tfrac{2}{3} & -\tfrac{4}{3} \\
\hline
\end{tabular}
};
\node at (8.2,0) [] {
\begin{tabular}{|C|CC|}
\hline
0 & -\tfrac{1}{3} & -\tfrac{1}{3} \\
\hline
\end{tabular}
};
\node at (6.75,-1.2) [] {$k=-\tfrac{4}{3}$};
\node at (5.5,1) [] {$\lambda_{r,s}$};
\node at (8,1) [] {$\Delta_{r,s}$};
\node at (5.5,-3.7) [] {
\begin{tabular}{|C|CC|}
\hline
0 & -\tfrac{4}{3} & -\tfrac{8}{3} \\
1 & -\tfrac{1}{3} & -\tfrac{5}{3} \\
2 & \tfrac{2}{3} & -\tfrac{2}{3} \\
\hline
\end{tabular}
};
\node at (8.2,-3.7) [] {
\begin{tabular}{|C|CC|}
\hline
0 & -\tfrac{1}{6} & \tfrac{1}{3} \\
\tfrac{9}{16} & -\tfrac{5}{48} & -\tfrac{5}{48} \\
\tfrac{3}{2} & \tfrac{1}{3} & -\tfrac{1}{6} \\
\hline
\end{tabular}
};
\node at (6.75,-5.5) [] {$k=-\tfrac{2}{3}$};
\node at (5.5,-2) [] {$\lambda_{r,s}$};
\node at (8,-2) [] {$\Delta_{r,s}$};
\node at (2.3,-7.3) [] {
\begin{tabular}{|C|CCC|}
\hline
0 & -\tfrac{3}{4} & -\tfrac{3}{2} & -\tfrac{9}{4} \\
1 & \tfrac{1}{4} & -\tfrac{1}{2} & -\tfrac{5}{4} \\
\hline
\end{tabular}
};
\node at (5.7,-7.3) [] {
\begin{tabular}{|C|CCC|}
\hline
0 & -\tfrac{5}{16} & -\tfrac{1}{4} & \tfrac{3}{16} \\
1 & \tfrac{3}{16} & -\tfrac{1}{4} & -\tfrac{5}{16} \\
\hline
\end{tabular}
};
\node at (4,-8.5) [] {$k=-\tfrac{5}{4}$};
\node at (2.3,-6.3) [] {$\lambda_{r,s}$};
\node at (5.7,-6.3) [] {$\Delta_{r,s}$};
\end{tikzpicture}
\caption{Tables of admissible highest weights $\lambda_{r,s}$ and their conformal dimensions $\Delta_{r,s}$ for certain admissible levels $k$.  The label $r$ runs from $1$ to $u-1$, increasing as one moves down, and $s$ runs from $0$ to $v-1$, increasing to the right.} \label{fig:Tables}
\end{center}
\end{figure}
}

Physically, this spectrum of admissible \hwms{} is not acceptable when $v>1$.  The reason is that, unlike the $\IrrMod{}$-type modules which are self-conjugate, the conjugates of the $\DiscMod{}^+$-type modules are not \hwms{}.  If we do not admit these conjugates in the spectrum, then the fields corresponding to the $\DiscMod{}^+$-type modules will necessarily vanish in all correlation functions.  We therefore conclude that, for $v>1$, the spectrum must be extended by the conjugate modules
\begin{equation}
\DiscMod{r,s}^- \equiv \conjmod{\DiscMod{r,s}^+} \qquad \text{(\(r=1,2,\ldots,u-1\); \(s=1,2,\ldots,v-1\)).}
\end{equation}
Just as the zero-grade subspace of $\DiscMod{r,s}^+$ may be identified with the (infinite-dimensional) highest weight $\SLA{sl}{2}$-module of highest weight $\lambda_{r,s}$, that of the conjugate module $\DiscMod{r,s}^-$ may be identified with the (infinite-dimensional) lowest weight $\SLA{sl}{2}$-module of lowest weight $-\lambda_{r,s}$.  Note that the $\DiscMod{r,s}^-$ are not lowest weight $\AKMA{sl}{2}$-modules.  They may, however, be regarded as \emph{relaxed} \hwms{}.

A relaxed \hwm{} is one that is generated by a relaxed \hws{}, this in turn being defined, for $\AKMA{sl}{2}$, as an eigenstate of $h_0$ which is annihilated by the modes $e_n$, $h_n$ and $f_n$, with $n>0$.  A standard \hws{} is therefore a relaxed \hws{} which also happens to be annihilated by $e_0$.  This terminology seems to have first appeared in \cite{FeiEqu98}, though such modules had been considered much earlier.  In particular, Adamovi\'{c} and Milas also determined the admissible $\AKMA{sl}{2}$-modules in the category of relaxed \hwms{}:

\begin{thm}[Adamovi\'{c}--Milas]
The admissible irreducibles from the category of \emph{relaxed} highest weight $\AKMA{sl}{2}$-modules at (admissible) level $k$ are precisely the admissible \hwms{}, their conjugates, and the following family of modules:
\begin{itemize}
\item $\TypMod{\lambda ; \Delta_{r,s}}$, for $r=1,2,\ldots,u-1$; $s=1,2,\ldots,v-1$ and $\lambda \in \RR / 2 \ZZ$ with $\lambda \neq \lambda_{r,s} , \lambda_{u-r,v-s} \bmod{2}$.
\end{itemize}
\end{thm}

\noindent Here, $\TypMod{\lambda ; \Delta_{r,s}}$ denotes the irreducible\footnote{The requirement that $\lambda \neq \lambda_{r,s} , \lambda_{u-r,v-s}$ stems from the fact that the modules $\TypMod{\lambda_{r,s} ; \Delta_{r,s}}$ and $\TypMod{\lambda_{u-r,v-s} ; \Delta_{r,s}}$ would not be irreducible.  We shall discuss the indecomposable modules that correspond to $\lambda = \lambda_{r,s} , \lambda_{u-r,v-s}$ in detail in \secref{sec:TypChar}.} relaxed \hwm{} whose zero-grade subspace is spanned by an infinite number of states, parametrised by $n \in \ZZ$, each of which has conformal dimension $\Delta_{r,s}$ and weight of the form $\lambda + 2n$.  This zero-grade subspace may be identified with an irreducible $\SLA{sl}{2}$-module of principal series type, meaning that it possesses neither a highest nor a lowest weight.  The $\AKMA{sl}{2}$-module $\TypMod{\lambda ; \Delta_{r,s}}$ may be constructed by appropriately inducing this $\SLA{sl}{2}$-module and taking the irreducible quotient.

The spectrum of irreducible admissibles therefore includes $u-1$ modules $\IrrMod{r,0}$, $\brac{u-1} \brac{v-1}$ modules $\DiscMod{r,s}^+$ and the same number of conjugate modules $\DiscMod{r,s}^-$, and $\tfrac{1}{2} \brac{u-1} \brac{v-1}$ continuous families of modules $\TypMod{\lambda ; \Delta_{r,s}}$ (because $\Delta_{r,s} = \Delta_{u-r,v-s}$ and there are no other coincidences of conformal dimensions).  Aside from the conjugation $\conjmod{\TypMod{\lambda ; \Delta_{r,s}}} = \TypMod{-\lambda ; \Delta_{r,s}}$, these admissibles are further related by spectral flow as follows:
\begin{equation} \label{eqn:SFRelations}
\sfmod{}{\IrrMod{r,0}} = \DiscMod{u-r,v-1}^+, \qquad \sfmod{-1}{\IrrMod{r,0}} = \DiscMod{u-r,v-1}^-, \qquad \sfmod{-1}{\DiscMod{r,s}^+} = \DiscMod{u-r,v-1-s}^- \quad \text{(\(s \neq v-1\)).}
\end{equation}
Of course, this has to be slightly adjusted in the non-negative integer level case:
\begin{equation}
\sfmod{}{\IrrMod{r,0}} = \sfmod{-1}{\IrrMod{r,0}} = \IrrMod{u-r,0} \qquad \text{(\(v=1\)).}
\end{equation}
Excluding this case, it makes sense to ask about modules obtained from higher spectral flows.  It turns out that for every $v>1$, the spectral flow images $\sfmod{\ell}{\mathcal{M}}$, $\ell \in \ZZ$, of any admissible module $\mathcal{M}$ are mutually non-isomorphic.  However, only three at most of these infinitely many images may be identified as relaxed \hwms{}.  The rest are irreducibles whose conformal dimensions \emph{are not bounded below}.  Nevertheless, these images are still admissible modules.\footnote{Spectral flow automorphisms do not, strictly speaking, define vertex algebra automorphisms because they do not preserve the vacuum.  However, they do preserve \opes{} which is enough to show that they map admissibles to admissibles.}  We therefore have to enlarge the spectrum of irreducible admissibles, this time for the last time.

To summarise, we may characterise the irreducible spectrum, for $v>1$, as consisting of:
\begin{itemize}[leftmargin=*]
\item $u-1$ countably-infinite families parametrised by $r=1,2,\ldots,u-1$:
\[
\sfmod{\ell}{\IrrMod{r,0}}, \qquad \text{(\(\ell \in \ZZ\)).}
\]
\item $\brac{u-1} \brac{v-2}$ countably-infinite families parametrised by $r=1,2,\ldots,u-1$ and $s=1,2,\ldots,v-2$:
\[
\sfmod{\ell}{\DiscMod{r,s}^+}, \qquad \text{(\(\ell \in \ZZ\)).}
\]
\item $\tfrac{1}{2} \brac{u-1} \brac{v-1}$ uncountably-infinite families parametrised by $r=1,2,\ldots,u-1$ and $s=1,2,\ldots,v-1$:
\[
\sfmod{\ell}{\TypMod{\lambda ; \Delta_{r,s}}}, \qquad \text{(\(\ell \in \ZZ\), \(\lambda \in \RR / 2 \ZZ\) and \(\lambda \neq \lambda_{r,s} , \lambda_{u-r,v-s} \bmod{2}\)).}
\]
\end{itemize}
We remark that the given range of $s$ for the second class of families is correct because of the isomorphisms $\sfmod{\ell}{\DiscMod{r,v-1}^+} \cong \sfmod{\ell+1}{\IrrMod{u-r,0}}$.  The three different types of families are illustrated in \figref{fig:Spec}.  Finally, we will refer to the $\sfmod{\ell}{\TypMod{\lambda ; \Delta_{r,s}}}$ as the \emph{standard modules} of the theory, following \cite{CreLog13}.  When a standard module $\sfmod{\ell}{\TypMod{\lambda ; \Delta_{r,s}}}$ is irreducible, which occurs whenever $\lambda \neq \lambda_{r,s}, \lambda_{u-r,v-s}$, we shall refer to it as being \emph{typical}.  Admissible modules which are not typical, such as the $\sfmod{\ell}{\IrrMod{r,0}}$ and the $\sfmod{\ell}{\DiscMod{r,s}^+}$, are said to be \emph{atypical}.

{
\psfrag{L1}[][]{$\IrrMod{r,0} \equiv \IrrMod{r-1}$}
\psfrag{E0}[][]{$\TypMod{\lambda ; \Delta_{r,s}}$}
\psfrag{Lb}[][]{$\DiscMod{u-r,v-1}^+$}
\psfrag{Da}[][]{$\DiscMod{r,s}^+$}
\psfrag{Lb*}[][]{$\DiscMod{u-r,v-1}^-$}
\psfrag{Da*}[][]{$\DiscMod{u-r,v-s-1}^-$}
\psfrag{g}[][]{$\sfaut$}
\psfrag{ee}[][]{}
\psfrag{ff}[][]{}
\psfrag{gg}[][]{}
\psfrag{hh}[][]{}
\psfrag{0e}[][]{}
\psfrag{aq}[][]{}
\psfrag{cq}[][]{}
\psfrag{bq}[][]{}
\psfrag{dq}[][]{}
\psfrag{ii}[][]{}
\psfrag{jj}[][]{}
\psfrag{kk}[][]{}
\psfrag{ll}[][]{}
\psfrag{mm}[][]{}
\begin{figure}
\begin{center}
\includegraphics[width=\textwidth]{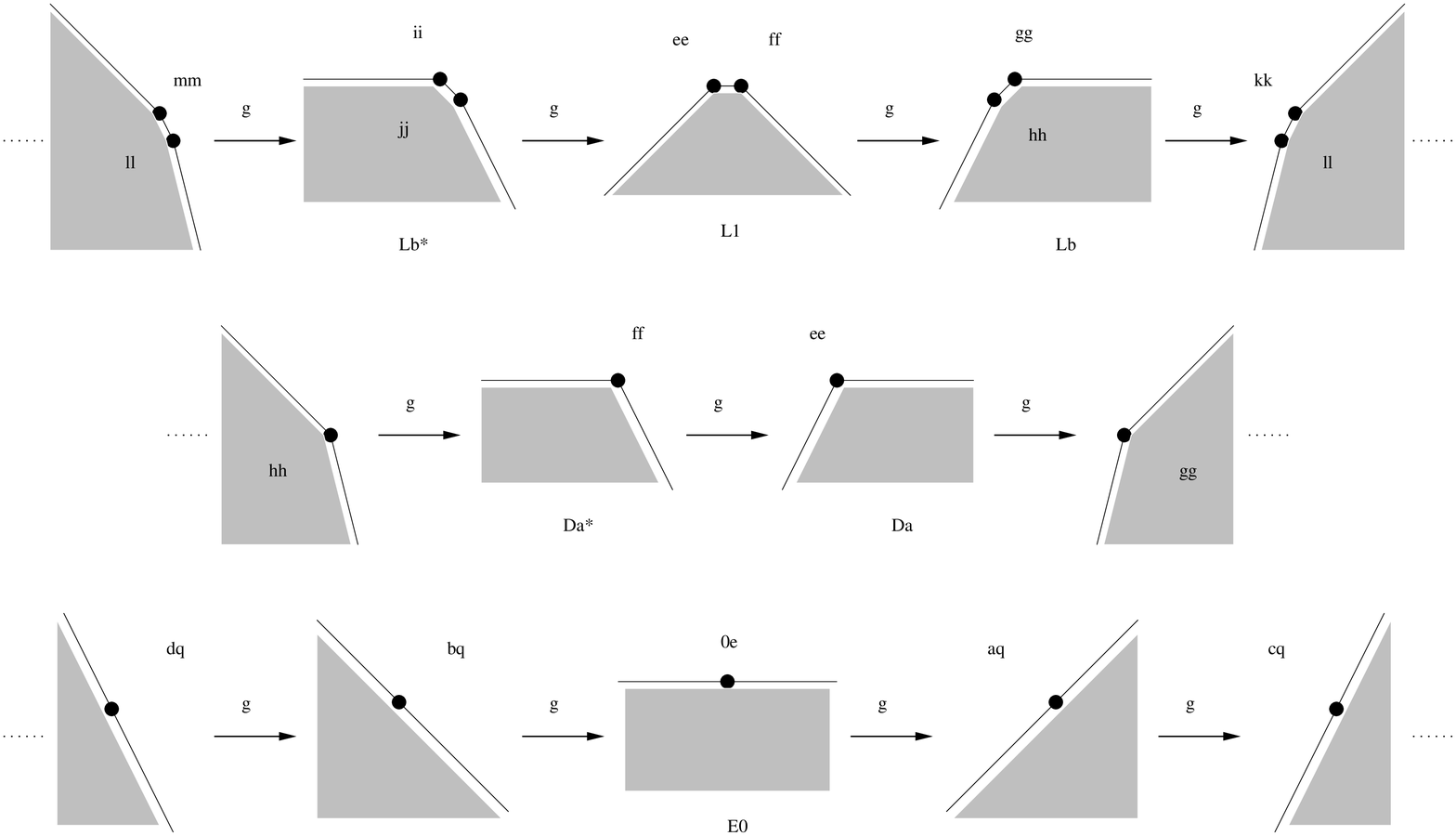}
\caption{Depictions of the three types of families of admissible irreducible $\AKMA{sl}{2}$-modules when $v>1$.  Conformal dimensions increase from top to bottom and $\SLA{sl}{2}$-weights increase from right to left.} \label{fig:Spec}
\end{center}
\end{figure}
}

\section{Standard Characters} \label{sec:TypChar}

We will assume, unless otherwise stipulated, that $v>1$ for the remainder of the article.  The admissible modules with $v=1$ coincide with the well-known integrable modules at non-negative integer level and we refer to standard texts, for example \cite{KacInf90,DiFCon97}, for their study.

To derive character formulae for the standard modules $\sfmod{\ell}{\TypMod{\lambda; \Delta_{r,s}}}$, it is actually convenient to start with certain atypical characters.  We therefore consider the structure of the Verma modules $\VerMod{r,s}$, for $r=1,2,\ldots,u-1$ and $s=1,2,\ldots,v-1$, whose level $k$ is admissible and whose irreducible quotients are the admissible modules $\DiscMod{r,s}^+$.  The characters of these Verma modules are simply given by
\begin{align}
\fch{\VerMod{r,s}}{y;z;q} &= \traceover{\VerMod{r,s}} y^k z^{h_0} q^{L_0 - c/24} = \frac{y^k z^{\lambda_{r,s}} q^{\Delta_{r,s} - c/24}}{\prod_{i=1}^{\infty} \brac{1 - z^2 q^i} \brac{1-q^i} \brac{1 - z^{-2} q^{i-1}}} \notag \\
&= \frac{-\ii y^k z^{\lambda_{r,s} + 1} q^{\Delta_{r,s} - c/24 + 1/8}}{\Jth{1}{z^2 ; q}}.
\end{align}
Their structures may be obtained straight-forwardly from the Kac-Kazhdan formula.  The singular vectors turn out to have weights of the form $\lambda_{r',s}$ and conformal dimensions $\Delta_{r',s}$, where $r' = \pm r \bmod{u}$.  More precisely, the singular vectors form an infinite braided pattern as follows:

\medskip
\begin{center}
\begin{tikzpicture}[auto,thick]
\node (O) at (0,0) [] {$r$};
\node (t1) at (1,1) [] {$-r$};
\node (b1) at (1,-1) [] {$2u-r$};
\node (t2) at (3,1) [] {$-2u+r$};
\node (b2) at (3,-1) [] {$2u+r$};
\node (t3) at (5,1) [] {$-2u-r$};
\node (b3) at (5,-1) [] {$4u-r$};
\node (t4) at (7,1) [] {$-4u+r$};
\node (b4) at (7,-1) [] {$4u+r$};
\node (t5) at (9,1) [] {$-4u-r$};
\node (b5) at (9,-1) [] {$6u-r$};
\draw [->] (O) to (t1);
\draw [->] (O) to (b1);
\draw [->] (t1) to (t2);
\draw [->] (t1) to (b2);
\draw [->] (b1) to (t2);
\draw [->] (b1) to (b2);
\draw [->] (t2) to (t3);
\draw [->] (t2) to (b3);
\draw [->] (b2) to (t3);
\draw [->] (b2) to (b3);
\draw [->] (t3) to (t4);
\draw [->] (t3) to (b4);
\draw [->] (b3) to (t4);
\draw [->] (b3) to (b4);
\draw [->] (t4) to (t5);
\draw [->] (t4) to (b5);
\draw [->] (b4) to (t5);
\draw [->] (b4) to (b5);
\draw [dotted] (t5) to (10.5,1);
\draw [dotted] (b5) to (10.5,-1);
\draw [dotted] (9.5,0) to (10.5,0);
\end{tikzpicture}
\end{center}
\medskip

\noindent Here, we indicate the singular vector by the value of $r'$, for clarity.  Adding and subtracting the characters of the Verma modules generated by these singular vectors, we arrive at a character formula for the $\DiscMod{r,s}^+$:

\begin{prop} \label{prop:ChDrs}
Let $k$ be an admissible level with $v>1$.  Then, for $r=1,2,\ldots,u-1$ and $s=1,2,\ldots,v-1$, the character of the irreducible admissible module $\DiscMod{r,s}^+$ is given by
\begin{equation} \label{mch:Drs}
\fch{\DiscMod{r,s}^+}{y;z;q} = \frac{-\ii y^k z^{\lambda_{r,s} + 1} q^{\Delta_{r,s} - c/24 + 1/8}}{\Jth{1}{z^2 ; q}} \sum_{j \in \ZZ} \sqbrac{z^{2uj} q^{j \brac{uvj+vr-us}} - z^{2 \brac{uj-r}} q^{\brac{uj-r} \brac{vj-s}}}.
\end{equation}
\end{prop}

\noindent The character of the conjugate module $\DiscMod{r,s}^- = \conjmod{\DiscMod{r,s}^+}$ is obtained from this formula by inverting $z$.

The zero-grade subspace of $\DiscMod{r,s}^+$ has a basis in which each basis state has weight of the form $\lambda_{r,s} - 2m$, $m=0,1,2,\ldots$, and conformal dimension $\Delta_{r,s}$.  By \eqnref{eqn:KacSymmetry}, the zero-grade subspace of $\DiscMod{u-r,v-s}^-$ has a similar basis in which the states have weights $\lambda_{r,s} + 2m$, $m=1,2,3,\ldots$, and conformal dimension $\Delta_{r,s}$.  It follows that there exist \emph{indecomposable} modules in which $\DiscMod{r,s}^+$ and $\DiscMod{u-r,v-s}^-$ are ``glued together'' by the action of $\AKMA{sl}{2}$.\footnote{The existence of these indecomposables may be demonstrated by inducing the indecomposable $\SLA{sl}{2}$-module corresponding to this zero-grade subspace and then quotienting by the maximal submodule among those having trivial intersection with the zero-grade subspace.}  In fact, there are two such non-isomorphic indecomposables:  One which has $\DiscMod{r,s}^+$ as a submodule and $\DiscMod{u-r,v-s}^-$ as the quotient by this submodule, and one for which the identities of the submodule and quotient are swapped.  We denote these atypical indecomposables by $\TypMod{r,s}^+$ and $\TypMod{u-r,v-s}^-$, respectively, and summarise their structure in the following short exact sequences:
\begin{equation} \label{ES:Ers}
\dses{\DiscMod{r,s}^+}{\TypMod{r,s}^+}{\DiscMod{u-r,v-s}^-}, \qquad
\dses{\DiscMod{r,s}^-}{\TypMod{r,s}^-}{\DiscMod{u-r,v-s}^+}.
\end{equation}
Note that the $\TypMod{r,s}^{\pm}$ and $\TypMod{u-r,v-s}^{\pm}$ correspond precisely to the ``holes'' in the continuous spectrum of the admissible irreducibles $\TypMod{\lambda ; \Delta_{r,s}}$.  These holes acknowledge the fact that the admissibles would fail to be irreducible if we were to allow $\lambda = \lambda_{r,s}$ or $\lambda = \lambda_{u-r,v-s} = -\lambda_{r,s} \bmod 2$.  We remark that $\conjmod{\TypMod{r,s}^+} = \TypMod{r,s}^-$.

Our next task is to compute the character of $\TypMod{r,s}^+$.  From \eqref{mch:Drs}, we easily obtain that of $\DiscMod{u-r,v-s}^-$:
\begin{multline}
\ch{\DiscMod{u-r,v-s}^-} = \frac{-\ii y^k z^{-\lambda_{u-r,v-s} - 1} q^{\Delta_{u-r,v-s} - c/24 + 1/8}}{\Jth{1}{z^{-2} ; q}} \\
\cdot \sum_{j \in \ZZ} \sqbrac{z^{-2uj} q^{j \brac{uvj-vr+us}} - z^{-2 \brac{u \brac{j-1} + r}} q^{\brac{u \brac{j-1} + r} \brac{v \brac{j-1} + s}}}.
\end{multline}
Using \eqref{eqn:KacSymmetry}, $\Jth{1}{z^{-2} ; q} = -\Jth{1}{z^2 ; q}$, and sending $j$ to $-j$ in the first term and $j$ to $-j+1$ in the second term of the sum, we find that this character is identical to $-\ch{\DiscMod{r,s}^+}$.  In other words,
\begin{equation} \label{eqn:ZeroCharacter}
\ch{\TypMod{r,s}^+} = \ch{\DiscMod{r,s}^+} + \ch{\DiscMod{u-r,v-s}^-} = 0.
\end{equation}
The character of the conjugate module $\TypMod{r,s}^-$ likewise vanishes.  These vanishings generalise the results obtained for $k=-\tfrac{1}{2}$ and $k=-\tfrac{4}{3}$ in \cite{RidSL210,CreMod12} (see also \cite{KohFus88}).

Of course, the fact that the characters of the $\TypMod{r,s}^{\pm}$ vanish does not mean that the modules vanish.  As emphasised in \cite{CreMod12}, it just means that we should not consider these characters as meromorphic functions of $z$, but rather as formal power series (or better yet, algebraic distributions).  The point is that the character formula for $\DiscMod{r,s}^+$ given in \propref{prop:ChDrs} is only valid (assuming $v>1$) when expanded in the region \cite{KacMod88,FeiRes98}
\begin{equation} \label{eqn:Convergence}
\abs{q} < 1, \qquad 
\begin{cases}
1 < \abs{z}^2 < \abs{q}^{-1} & \text{(\(s \neq v-1\)),} \\
1 < \abs{z}^2 < \abs{q}^{-2} & \text{(\(s = v-1\)).}
\end{cases}
\end{equation}
The corresponding region for $\ch{\DiscMod{r,s}^-}$ is obtained by inverting $z$, so we immediately see that the regions for $\ch{\DiscMod{r,s}^+}$ and $\ch{\DiscMod{u-r,v-s}^-}$ are \emph{disjoint} \cite{FeiRes98,LesSU202}, hence that the sum \eqref{eqn:ZeroCharacter} is invalid when the characters are expanded as power series.  In fact, what this tells us is that these characters only sum to zero upon meromorphically extending them to the entire $z$-plane.

To correctly compute the sum of $\ch{\DiscMod{r,s}^+}$ and $\ch{\DiscMod{u-r,v-s}^-}$, and thereby obtain the character of $\TypMod{r,s}^+$, we use expansion formulae derived in \cite{KacInt94}, as explained in \cite[App.~A]{CreMod12}:
\begin{subequations}
\begin{align}
\frac{1}{\Jth{1}{z^2;q}} &= \frac{-\ii \Jth{1}{w;q}}{\Jth{1}{wz^2;q} \func{\eta}{q}^3} \sum_{n \in \ZZ} \frac{wz^{2n}q^n}{1-wq^n} & &\text{($1 < \abs{z}^2 < \abs{q}^{-1}$),} \label{eqn:KW1} \\
\frac{1}{\Jth{1}{z^2;q}} &= \frac{-\ii \Jth{1}{w;q}}{\Jth{1}{wz^2;q} \func{\eta}{q}^3} \sum_{n \in \ZZ} \frac{z^{2n}}{1-wq^n} & &\text{($\abs{q} < \abs{z}^2 < 1$).} \label{eqn:KW2}
\end{align}
\end{subequations}
Applying \eqref{eqn:KW1} to $\ch{\DiscMod{r,s}^+}$ and \eqref{eqn:KacSymmetry}, $\Jth{1}{z^{-2} ; q} = -\Jth{1}{z^2 ; q}$, and \eqref{eqn:KW2} to $\ch{\DiscMod{r,s}^-}$, the sum of the characters becomes
\begin{equation}
\ch{\TypMod{r,s}^+} = \frac{y^k z^{\lambda_{r,s} + 1} q^{\Delta_{r,s} - c/24 + 1/8}}{\func{\eta}{q}^3} \frac{\Jth{1}{w;q}}{\Jth{1}{wz^2;q}} \sum_{n \in \ZZ} z^{2n} \sum_{j \in \ZZ} \sqbrac{z^{2uj} q^{j \brac{uvj+vr-us}} - z^{2 \brac{uj-r}} q^{\brac{uj-r} \brac{vj-s}}}.
\end{equation}
This can be dramatically simplified by writing $z = \ee^{2 \pi \ii \zeta}$ and employing the identity
\begin{equation}
\sum_{n \in \ZZ} \ee^{4 \pi \ii \zeta n} = \sum_{m \in \ZZ} \func{\delta}{2 \zeta - m},
\end{equation}
valid as an equality of (algebraic) distributions.  Because $\Jth{1}{\ee^{2 \pi \ii m} w;q} = \ee^{\ii \pi m} \Jth{1}{w;q}$ for $m \in \ZZ$, our character sum becomes
\begin{align}
\ch{\TypMod{r,s}^+} &= \frac{y^k z^{\lambda_{r,s} + 1} q^{\Delta_{r,s} - c/24 + 1/8}}{\func{\eta}{q}^3} \sum_{m \in \ZZ} \func{\delta}{2 \zeta - m} \ee^{-\ii \pi m} \sum_{j \in \ZZ} \sqbrac{q^{j \brac{uvj+vr-us}} - q^{\brac{uj-r} \brac{vj-s}}} \notag \\
&= \frac{y^k z^{\lambda_{r,s}}}{\func{\eta}{q}^2} \frac{q^{\Delta_{r,s}^{\mathrm{Vir}} - c^{\mathrm{Vir}}/24 + 1/24} \sum_{j \in \ZZ} \sqbrac{q^{j \brac{uvj+vr-us}} - q^{\brac{uj-r} \brac{vj-s}}}}{\func{\eta}{q}} \sum_{n \in \ZZ} z^{2n}.
\end{align}
Here, we have used \eqref{eqn:ModularMagic} to express $\Delta_{r,s}$ and $c$ in terms of their Virasoro analogues because we recognise the second factor above (see \cite{RocVac85}).

\begin{prop} \label{prop:ChErs}
Let $k$ be an admissible level with $v>1$.  Then, for $r=1,2,\ldots,u-1$ and $s=1,2,\ldots,v-1$, the character of the indecomposable admissible module $\TypMod{r,s}^+$ is given by
\begin{equation} \label{ch:Ers}
\fch{\TypMod{r,s}^+}{y;z;q} = \frac{y^k z^{\lambda_{r,s}} \fvch{r,s}{q}}{\func{\eta}{q}^2} \sum_{n \in \ZZ} z^{2n},
\end{equation}
where $\vch{r,s}$ denotes the character of the irreducible Virasoro module whose \hws{} has conformal dimension $\Delta_{r,s}^{\mathrm{Vir}}$.
\end{prop}

\noindent The character of $\TypMod{r,s}^-$ is obtained by conjugating (inverting $z$) and one easily sees that $\ch{\TypMod{r,s}^-} = \ch{\TypMod{u-r,v-s}^+}$.  This proposition ties the characters of the indecomposables $\TypMod{r,s}^{\pm}$ to those of the Virasoro minimal model $\minmod{u}{v}$, strengthening an old observation of Mukhi and Panda \cite{MukFra90}.  Moreover, the structure theory for relaxed \hwms{} (see \cite{FeiEqu98,SemEmb97}) allows us to conclude something even stronger:

\begin{cor} \label{cor:ChTyp}
Let $k$ be an admissible level with $v>1$.  Then, the character of the irreducible admissible module $\TypMod{\lambda ; \Delta_{r,s}}$ is given by
\begin{equation} \label{ch:Typ}
\fch{\TypMod{\lambda ; \Delta_{r,s}}}{y;z;q} = \frac{y^k z^{\lambda} \fvch{r,s}{q}}{\func{\eta}{q}^2} \sum_{n \in \ZZ} z^{2n}.
\end{equation}
\end{cor}

\noindent The action of spectral flow upon the character of an arbitrary $\AKMA{sl}{2}$-module $\mathcal{M}$,
\begin{equation} \label{eqn:CharSF}
\fch{\sfmod{\ell}{\mathcal{M}}}{y;z;q} = \fch{\mathcal{M}}{y z^{\ell} q^{\ell^2 / 4} ; z q^{\ell / 2} ; q},
\end{equation}
may then be used to obtain the characters of the remaining standard modules $\sfmod{\ell}{\TypMod{\lambda ; \Delta_{r,s}}}$.

\begin{ex}
The characters of the standard admissibles were worked out for $k=-\tfrac{1}{2}$ in \cite{RidSL210}, using the fact that the $\beta \gamma$ ghost system is a free field theory, and for $k=-\tfrac{4}{3}$ in \cite{CreMod12}, using the above method.  For both levels, the standard characters took the deceptively simple form
\begin{equation}
\ch{\TypMod{\lambda ; \Delta}} = \frac{y^k z^{\lambda}}{\func{\eta}{q}^2} \sum_{n \in \ZZ} z^{2n}.
\end{equation}
We can now understand this simplicity as resulting from the fact that $u=3$, $v=2$ for $k=-\tfrac{1}{2}$ and $u=2$, $v=3$ for $k=-\tfrac{4}{3}$.  The corresponding minimal model is, in both cases, the trivial theory $\minmod{2}{3} = \minmod{3}{2}$, so the Virasoro character appearing in the standard characters is just the constant $1$.  We also see that these are the only levels for which the standard characters are so simple.
\end{ex}

\section{Modular Transformations for Standard Characters} \label{sec:ModTyp}

The remarkable appearance of Virasoro minimal model characters in the standard character formulae is quite fortuitous, because it greatly facilitates the determination of the modular transformations.  For this, we write
\begin{equation}
y = \ee^{2 \pi \ii \theta}, \qquad z = \ee^{2 \pi \ii \zeta}, \qquad q = \ee^{2 \pi \ii \tau}
\end{equation}
and consider the effect on the characters of applying the standard S- and T-transformations
\begin{equation}
\modS \colon \modarg{\theta}{\zeta}{\tau} \longmapsto \modarg{\theta - \zeta^2 / \tau}{\zeta / \tau}{-1 / \tau}, \qquad
\modT \colon \modarg{\theta}{\zeta}{\tau} \longmapsto \modarg{\theta}{\zeta}{\tau + 1},
\end{equation}
for which $\modS^4 = \brac{\modS \modT}^6 = \id$.  We denote the action on characters by $\modS \set{\cdot}$ and $\modT \set{\cdot}$.

\begin{thm} \label{thm:TypMod}
Let $k$ be an admissible level with $v>1$.  Then, the characters of the standard admissible modules carry a (projective) representation of the modular group $\SLG{SL}{2 ; \ZZ}$.  Explicitly, the S-transformation is
\begin{subequations} \label{eqn:STyp}
\begin{equation} \label{eqn:STypAnsatz}
\modS \set{\ch{\sfmod{\ell}{\TypMod{\lambda ; \Delta_{r,s}}}}} = \sum_{\ell' \in \ZZ} \sideset{}{'} \sum_{r',s'} \int_{-1}^1 \modS_{(\ell,\lambda ; \Delta_{r,s}) (\ell',\lambda' ; \Delta_{r',s'})} \ch{\sfmod{\ell'}{\TypMod{\lambda' ; \Delta_{r',s'}}}} \: \dd \lambda',
\end{equation}
where the S-matrix entries are given by
\begin{equation} \label{eqn:STypEntries}
\modS_{(\ell,\lambda ; \Delta_{r,s}) (\ell',\lambda' ; \Delta_{r',s'})} = \frac{1}{2} \frac{\abs{\tau}}{-\ii \tau} \ee^{-\ii \pi \brac{k \ell \ell' + \ell \lambda' + \ell' \lambda}} \vmodS_{(r,s) (r',s')}
\end{equation}
and the $\minmod{u}{v}$ S-matrix entries are given, as usual, by \cite{CarOpe86,ItzTwo86}
\begin{equation} \label{eqn:SVir}
\vmodS_{(r,s) (r',s')} = -2 \sqrt{\frac{2}{uv}} \brac{-1}^{rs'+r's} \sin \frac{v \pi r r'}{u} \sin \frac{u \pi s s'}{v}.
\end{equation}
\end{subequations}
The prime on the sum in \eqref{eqn:STypAnsatz} indicates that $r'$ and $s'$ run over the entries of the Kac table of $\minmod{u}{v}$ modulo the Kac symmetry $\brac{r,s} \sim \brac{u-r,v-s}$.  The T-transformation is
\begin{subequations} \label{eqn:TTyp}
\begin{equation} \label{eqn:TTypAnsatz}
\modT \set{\ch{\sfmod{\ell}{\TypMod{\lambda ; \Delta_{r,s}}}}} = \ee^{\ii \pi \ell \brac{\lambda + k \ell / 2}} \ee^{2 \pi \ii \brac{\Delta_{r,s} - c/24}} \ch{\sfmod{\ell}{\TypMod{\lambda ; \Delta_{r,s}}}},
\end{equation}
so the T-matrix entries are given by
\begin{equation} \label{eqn:TTypEntries}
\modT_{(\ell,\lambda ; \Delta_{r,s}) (\ell',\lambda' ; \Delta_{r',s'})} = \ee^{\ii \pi \ell \brac{\lambda + k \ell / 2}} \ee^{2 \pi \ii \brac{\Delta_{r,s} - c/24}} \delta_{\ell = \ell'} \func{\delta}{\lambda = \lambda' \bmod{2}} \delta_{r=r'} \delta_{s=s'},
\end{equation}
\end{subequations}
again for $\brac{r,s}$ and $\brac{r',s'}$ restricted by Kac symmetry.
\end{thm}
\begin{proof}
We begin by rewriting the standard characters \eqref{ch:Typ} as functions of $\theta$, $\zeta$ and $\tau$:
\begin{equation}
\ch{\TypMod{\lambda ; \Delta_{r,s}}} = \frac{\ee^{2 \pi \ii k \theta} \fvch{r,s}{\tau}}{\func{\eta}{\tau}^2} \sum_{m \in \ZZ} \ee^{\ii \pi m \lambda} \func{\delta}{2 \zeta - m}.
\end{equation}
Applying spectral flow, as in \eqref{eqn:CharSF}, we obtain the general character formula
\begin{align} \label{ch:TypSF}
\ch{\sfmod{\ell}{\TypMod{\lambda ; \Delta_{r,s}}}} &= \frac{\ee^{2 \pi \ii k \brac{\theta + \ell \zeta + \ell^2 \tau / 4}} \fvch{r,s}{\tau}}{\func{\eta}{\tau}^2} \sum_{m \in \ZZ} \ee^{\ii \pi m \lambda} \func{\delta}{2 \zeta + \ell \tau - m} \notag \\
&= \frac{\ee^{2 \pi \ii k \theta} \ee^{-\ii \pi k \ell^2 \tau / 2} \fvch{r,s}{\tau}}{\func{\eta}{\tau}^2} \sum_{m \in \ZZ} \ee^{\ii \pi m \brac{\lambda + k \ell}} \func{\delta}{2 \zeta + \ell \tau - m}.
\end{align}
Sending $\tau$ to $\tau + 1$ and using the known transformation properties of $\vch{r,s}$ and $\eta$, it is now straight-forward to arrive at the T-transformation \eqref{eqn:TTyp}.

Verifying the S-transformation requires a little more work.  First, we apply $\modS$ to \eqref{ch:TypSF}:
\begin{align} \label{eqn:Goal}
\modS &\set{\ch{\sfmod{\ell}{\TypMod{\lambda ; \Delta_{r,s}}}}} = \frac{\ee^{2 \pi \ii k \brac{\theta - \zeta^2 / \tau}} \ee^{\ii \pi k \ell^2 / 2 \tau} \fvch{r,s}{-1 / \tau}}{\func{\eta}{-1 / \tau}^2} \sum_{m \in \ZZ} \ee^{\ii \pi m \brac{\lambda + k \ell}} \func{\delta}{\frac{2 \zeta - \ell - m \tau}{\tau}} \notag \\
&= \frac{\ee^{2 \pi \ii k \theta} \ee^{\ii \pi k \ell^2 / 2 \tau}}{-\ii \tau \: \func{\eta}{\tau}^2} \sideset{}{'} \sum_{r',s'} \vmodS_{(r,s) (r',s')} \fvch{r',s'}{\tau} \cdot \sum_{m \in \ZZ} \ee^{\ii \pi m \brac{\lambda + k \ell}} \ee^{-\ii \pi k \brac{\ell + m \tau}^2 / 2 \tau} \abs{\tau} \func{\delta}{2 \zeta - \ell - m \tau} \notag \\
&= \frac{\abs{\tau}}{-\ii \tau} \frac{\ee^{2 \pi \ii k \theta}}{\func{\eta}{\tau}^2} \sideset{}{'} \sum_{r',s'} \vmodS_{(r,s) (r',s')} \fvch{r',s'}{\tau} \sum_{m \in \ZZ} \ee^{-\ii \pi m \lambda} \ee^{-\ii \pi k m^2 \tau / 2} \func{\delta}{2 \zeta + m \tau - \ell}.
\end{align}
Now substitute \eqref{eqn:STypEntries} and \eqref{ch:TypSF} (with $\ell$, $\lambda$, $r$ and $s$ replaced by their primed counterparts) into \eqref{eqn:STypAnsatz}.  The integral over $\lambda'$ is easy to evaluate, resulting in $\delta_{\ell = m}$, and this then allows one to perform the sum over $m$.  Simplifying, and relabelling $\ell'$ as $m$, we recover \eqref{eqn:Goal}.
\end{proof}

\noindent We will refer to the quantities $\modS_{(\ell,\lambda ; \Delta_{r,s}) (\ell',\lambda' ; \Delta_{r',s'})}$ and $\modT_{(\ell,\lambda ; \Delta_{r,s}) (\ell',\lambda' ; \Delta_{r',s'})}$ as \emph{matrix elements}, even though $\lambda$ and $\lambda'$ parametrise a continuous range.  Note that the integration range of $\lambda'$ in \eqref{eqn:STypAnsatz} is only required to be a fundamental domain for $\RR / 2 \ZZ$ and that any other interval of length $2$ would suffice.

\begin{cor} \label{cor:TypModProps}
The S- and T-matrices of \thmref{thm:TypMod} are symmetric and unitary.  Moreover,
\begin{equation}
\modS_{(-\ell,-\lambda ; \Delta_{r,s}) (-\ell',-\lambda' ; \Delta_{r',s'})} = \modS_{(\ell,\lambda ; \Delta_{r,s}) (\ell',\lambda' ; \Delta_{r',s'})}
\end{equation}
and
\begin{equation}
\brac{\modS^2}_{(\ell,\lambda ; \Delta_{r,s}) (\ell',\lambda' ; \Delta_{r',s'})} = \frac{\abs{\tau}^2}{-\tau^2} \delta_{\ell' = -\ell} \func{\delta}{\lambda' = -\lambda \bmod{2}} \delta_{r' = r} \delta_{s' = s}.
\end{equation}
$\modS^2$ is therefore conjugation at the level of characters, up to a phase.
\end{cor}

It is important to note that the representation of $\SLG{SL}{2 ; \ZZ}$ on the standard characters is only projective because of the phase $\abs{\tau} / \ii \tau$ appearing in \eqref{eqn:STypEntries}.  One can easily check that $\modS^4$ and $\brac{\modS \modT}^6$ are proportional to the identity transformation in this representation, with the proportionality constants being this phase to the fourth and sixth powers, respectively.  The fact that the S-matrix entries contain an explicit $\tau$-dependence through this phase is not really worrying because, as was explained in \cite{CreRel11}, this phase will cancel when pairing chiral and antichiral components to form (bulk) modular invariants.  Likewise, $\modS^2$ and conjugation differ by a phase for chiral modules, but are identical in the bulk.  Most importantly, this phase will also cancel when applying the Verlinde formula (\secref{sec:Verlinde}).

\section{Atypical Characters} \label{sec:AtypChar}

In this section, we return to the determination of the characters of the atypical irreducible admissibles.  While we have already computed the characters of the $\DiscMod{r,s}^+$ in \propref{prop:ChDrs}, and the remaining irreducible characters follow from applying spectral flow, there is still the issue of disjoint convergence regions to deal with.  Instead of revisiting this, we shall instead employ a well-known trick \cite{RozSTM93} in which atypical characters are computed as (infinite) linear combinations of limits of standard characters.  In this formalism, a (topological) basis for the linear span of the admissible characters is provided by those of the standard modules, recalling that these include the atypical indecomposables $\sfmod{\ell}{\TypMod{r,s}^+}$ and $\sfmod{\ell}{\TypMod{u-r,v-s}^+}$.

To verify that the atypical characters may indeed be expressed as (infinite) linear combinations of elements from this character basis, we follow \cite{CreRel11} in constructing resolutions for the atypical modules in terms of the indecomposables $\sfmod{\ell}{\TypMod{r,s}^+}$ and $\sfmod{\ell}{\TypMod{u-r,v-s}^+}$.  These are (infinite) exact sequences whose terms are all indecomposables of this form, except for the last two which are the atypical module and the zero module (in that order).  Their construction follows easily from repeatedly splicing the short exact sequences \eqref{ES:Ers} for the indecomposables with their spectral flow versions.

To begin, we apply spectral flow to the first sequence of \eqref{ES:Ers} so as to get $\DiscMod{r,s}^+$ as the quotient.  This preserves exactness.  However, because of the slight difference in the spectral flow orbit structures (see \eqref{eqn:SFRelations}), the results differ according as to whether $s=v-1$ or not:
\begin{subequations} \label{es:LD}
\begin{align}
\dses{\sfmod{}{\DiscMod{r,s+1}^+}}{\sfmod{}{\TypMod{r,s+1}^+}}{\DiscMod{r,s}^+}& & &\text{($s \neq v-1$),} \label{es:Drs} \\
\dses{\sfmod{2}{\DiscMod{u-r,1}^+}}{\sfmod{2}{\TypMod{u-r,1}^+}}{\DiscMod{r,v-1}^+}& & &\text{($s = v-1$).}
\end{align}
Since $\DiscMod{r,v-1}^+ \cong \sfmod{}{\IrrMod{u-r,0}}$, we also obtain
\begin{equation} \label{es:Lr}
\dses{\sfmod{}{\DiscMod{r,1}^+}}{\sfmod{}{\TypMod{r,1}^+}}{\IrrMod{r,0}}.
\end{equation}
\end{subequations}
Splicing these short exact sequences, we arrive at the desired resolutions.

\begin{prop} \label{prop:Resolutions}
Let $k$ be an admissible level with $v>1$.  Then, the atypical irreducible module $\IrrMod{r,0} = \IrrMod{r-1}$ has the following resolution:
\begin{subequations} \label{eqn:LES}
\begin{multline} \label{eqn:ResL}
\cdots \lra \sfmod{3v-1}{\TypMod{r,v-1}^+} \lra \cdots \lra \sfmod{2v+2}{\TypMod{r,2}^+} \lra \sfmod{2v+1}{\TypMod{r,1}^+} \\
\lra \sfmod{2v-1}{\TypMod{u-r,v-1}^+} \lra \cdots \lra \sfmod{v+2}{\TypMod{u-r,2}^+} \lra \sfmod{v+1}{\TypMod{u-r,1}^+} \\
\lra \sfmod{v-1}{\TypMod{r,v-1}^+} \lra \cdots \lra \sfmod{2}{\TypMod{r,2}^+} \lra \sfmod{}{\TypMod{r,1}^+} \lra \IrrMod{r,0} \lra 0.
\end{multline}
%Equivalently, this can be expressed as the following long, but finite-length, exact sequence:
%\begin{equation}
%0 \lra \sfmod{v}{\IrrMod{u-r,0}} \lra \sfmod{v-1}{\TypMod{r,v-1}^+} \lra \cdots \lra \sfmod{2}{\TypMod{r,2}^+} \lra \sfmod{}{\TypMod{r,1}^+} \lra \IrrMod{r,0} \lra 0.
%\end{equation}
For $\DiscMod{r,s}^+$, %a convenient long exact sequence is
%\begin{equation} \label{res:Drs}
%0 \lra \sfmod{v-s}{\IrrMod{u-r,0}} \lra \sfmod{v-1-s}{\TypMod{r,v-1}^+} \lra \cdots \lra \sfmod{2}{\TypMod{r,s+2}^+} \lra \sfmod{}{\TypMod{r,s+1}^+} \lra \DiscMod{r,s}^+ \lra 0,
%\end{equation}
%which may be spliced with \eqref{eqn:ResL} to obtain a resolution.
we have, for $s \neq v-1$, instead
\begin{multline}
\cdots \lra \sfmod{3v-s-1}{\TypMod{r,v-1}^+} \lra \cdots \lra \sfmod{2v-s+2}{\TypMod{r,2}^+} \lra \sfmod{2v-s+1}{\TypMod{r,1}^+} \\
\lra \sfmod{2v-s-1}{\TypMod{u-r,v-1}^+} \lra \cdots \lra \sfmod{v-s+2}{\TypMod{u-r,2}^+} \lra \sfmod{v-s+1}{\TypMod{u-r,1}^+} \\
\lra \sfmod{v-s-1}{\TypMod{r,v-1}^+} \lra \cdots \lra \sfmod{2}{\TypMod{r,s+2}^+} \lra \sfmod{}{\TypMod{r,s+1}^+} \lra \DiscMod{r,s}^+ \lra 0.
\end{multline}
%When $s=v-1$, the resolution is instead
%\begin{multline}
%\cdots \lra \sfmod{3v}{\TypMod{u-r,v-1}^+} \lra \cdots \lra \sfmod{2v+3}{\TypMod{u-r,2}^+} \lra \sfmod{2v+2}{\TypMod{u-r,1}^+} \\
%\lra \sfmod{2v}{\TypMod{r,v-1}^+} \lra \cdots \lra \sfmod{v+3}{\TypMod{r,2}^+} \lra \sfmod{v+2}{\TypMod{r,1}^+} \\
%\lra \sfmod{v}{\TypMod{u-r,v-1}^+} \lra \cdots \lra \sfmod{3}{\TypMod{u-r,2}^+} \lra \sfmod{2}{\TypMod{u-r,1}^+} \lra \DiscMod{r,v-1}^+ \lra 0.
%\end{multline}
(The resolution for $s=v-1$ may be obtained from \eqref{eqn:ResL} by applying spectral flow.)
\end{subequations}
\end{prop}

\noindent We remark that it is easy to derive similar resolutions involving the $\sfmod{\ell}{\TypMod{r,s}^-}$.  However, the character identity $\ch{\sfmod{\ell}{\TypMod{r,s}^+}} = \ch{\sfmod{\ell}{\TypMod{u-r,v-s}^-}}$ implies that this will not lead to anything new.

These resolutions are impressive, but they are really just a means of combining all the information contained in the short exact sequences \eqref{es:LD}.  For the $\IrrMod{r,0}$, this turns out to be very convenient for streamlining the computations of the following sections, but for the $\DiscMod{r,s}^+$, the above resolution is a little complicated and we will find it more convenient to work directly with the short exact sequences.  With this in mind, we present the corresponding character identities that we will need in what follows.

\begin{cor} \label{cor:AtypChars}
Let $k$ be an admissible level with $v>1$.  Then, the characters of the irreducible atypical modules are related to the characters of the indecomposable standard modules as follows:
\begin{subequations} \label{eqn:AtypChars}
\begin{align}
\ch{\IrrMod{r,0}} &= \sum_{s=1}^{v-1} \brac{-1}^{s-1} \sum_{\ell = 0}^{\infty} \set{\ch{\sfmod{2v \ell + s}{\TypMod{r,s}^+}} - \ch{\sfmod{2v \brac{\ell + 1} - s}{\TypMod{u-r,v-s}^+}}}, \label{eqn:ChL} \\
%&= \sum_{s=1}^{v-1} \brac{-1}^{s-1} \ch{\sfmod{s}{\TypMod{r,s}^+}} + \brac{-1}^{v-1} \ch{\sfmod{v}{\IrrMod{u-r,0}}}, \label{eqn:ChL'} \\
%\ch{\DiscMod{r,s}^+} &= \sum_{s'=s+1}^{v-1} \brac{-1}^{s'-s-1} \ch{\sfmod{s'-s}{\TypMod{r,s'}^+}} + \brac{-1}^{v-1-s} \ch{\sfmod{v-s}{\IrrMod{u-r,0}}}, \label{eqn:ChD}
\ch{\DiscMod{r,s}^+} &= \ch{\sfmod{}{\TypMod{r,s+1}^+}} - \ch{\sfmod{}{\DiscMod{r,s+1}^+}} \qquad \text{(\(s \neq v-1\)).} \label{eqn:Induct}
\end{align}
\end{subequations}
Of course, spectral flow may be used to obtain the expressions for the remaining atypicals.
\end{cor}

\noindent Note that the character formula \eqref{eqn:ChL} is convergent in the sense that the multiplicity of each weight space only receives contributions from finitely many of the characters of the indecomposables $\sfmod{\ell}{\TypMod{r,s}^+}$ and $\sfmod{\ell}{\TypMod{u-r,v-s}^+}$.

\begin{ex}
When $k=-\tfrac{1}{2}$, so $u=3$ and $v=2$, there are two atypical spectral flow orbits which we may take to be represented by $\IrrMod{0}$ and $\IrrMod{1}$.  The resolutions of \propref{prop:Resolutions} are
\begin{equation}
\begin{aligned}
\cdots \lra \sfmod{7}{\TypMod{2,1}^+} \lra \sfmod{5}{\TypMod{1,1}^+} \lra \sfmod{3}{\TypMod{2,1}^+} \lra \sfmod{}{\TypMod{1,1}^+} \lra \IrrMod{0} \lra &0, \\
\cdots \lra \sfmod{7}{\TypMod{1,1}^+} \lra \sfmod{5}{\TypMod{2,1}^+} \lra \sfmod{3}{\TypMod{1,1}^+} \lra \sfmod{}{\TypMod{2,1}^+} \lra \IrrMod{1} \lra &0
\end{aligned}
\end{equation}
and the character formulae from \corref{cor:AtypChars} become
\begin{equation}
\begin{aligned}
\ch{\IrrMod{0}} &= \sum_{\ell = 0}^{\infty} \set{\ch{\sfmod{4 \ell + 1}{\TypMod{1,1}^+}} - \ch{\sfmod{4 \ell + 3}{\TypMod{2,1}^+}}} = \sum_{\ell = 0}^{\infty} \brac{-1}^{\ell} \ch{\sfmod{2 \ell + 1}{\TypMod{\ell + 1/2 ; -1/8}^+}}, \\
\ch{\IrrMod{1}} &= \sum_{\ell = 0}^{\infty} \set{\ch{\sfmod{4 \ell + 1}{\TypMod{2,1}^+}} - \ch{\sfmod{4 \ell + 3}{\TypMod{1,1}^+}}} = \sum_{\ell = 0}^{\infty} \brac{-1}^{\ell} \ch{\sfmod{2 \ell + 1}{\TypMod{\ell - 1/2 ; -1/8}^+}}.
\end{aligned}
\end{equation}
Here, we note that $\lambda_{1,1} = -\tfrac{3}{2}$, $\lambda_{2,1} = -\tfrac{1}{2}$ and $\Delta_{1,1} = \Delta_{2,1} = -\tfrac{1}{8}$ (see \figref{fig:Tables}).  In this way, we recover the special case considered in \cite[Sec.~3.2]{CreMod12}.
\end{ex}

\section{Modular Transformations for Atypical Characters} \label{sec:ModAtyp}

It is now relatively easy to obtain the modular transformations of the characters of the $\sfmod{\ell}{\IrrMod{r,0}}$ from \corref{cor:AtypChars}.  The main difficulty is in simplifying the finite sums which appear in the character formula \eqref{eqn:ChL}.  This step may be overcome by using the following identity whose proof is straight-forward and best left to symbolic algebra packages.

\begin{lem} \label{lem:UsefulIdentity}
Given $R,S,u \in \ZZ$ and $v \in \ZZ \setminus \set{0}$, we have the following identity of functions of $\mu \in \RR$:
\begin{equation}
\brac{\cos \brac{\pi \mu} + \brac{-1}^R \cos \frac{u \pi S}{v}} \sum_{s=1}^{v-1} \brac{-1}^{\brac{R-1} \brac{s-1}} \sin \bigl( \brac{v-s} \pi \mu \bigr) \sin \frac{u \pi s S}{v} = \frac{1}{2} \sin \brac{v \pi \mu} \sin \frac{u \pi S}{v}.
\end{equation}
\end{lem}

\begin{thm} \label{thm:AtypMod}
Let $k$ be an admissible level with $v>1$.  Then, the characters of the atypical irreducible modules $\sfmod{\ell}{\IrrMod{r,0}} = \sfmod{\ell}{\IrrMod{r-1}}$ have the following S-transformations:
\begin{subequations} \label{eqn:SAtyp}
\begin{equation}
\modS \set{\ch{\sfmod{\ell}{\IrrMod{r,0}}}} = \sum_{\ell' \in \ZZ} \sideset{}{'} \sum_{r',s'} \int_{-1}^1 \modS_{\overline{(\ell;r,0)} (\ell',\lambda' ; \Delta_{r',s'})} \ch{\sfmod{\ell'}{\TypMod{\lambda' ; \Delta_{r',s'}}}} \: \dd \lambda'.
\end{equation}
Here, the atypical S-matrix entries are given by
\begin{equation} \label{eqn:SAtypEntries}
\modS_{\overline{(\ell;r,0)} (\ell',\lambda' ; \Delta_{r',s'})} = \frac{1}{2} \frac{\abs{\tau}}{-\ii \tau} \frac{\ee^{-\ii \pi \brac{k \ell \ell' + \ell \lambda' + \ell' \brac{r-1}}}}{2 \cos \brac{\pi \lambda'} + \brac{-1}^{r'} 2 \cos \brac{k \pi s'}} \vmodS_{(r,1) (r',s')}
\end{equation}
\end{subequations}
and the $\minmod{u}{v}$ S-matrix entries were given in \eqnref{eqn:SVir}.
\end{thm}
\begin{proof}
From \eqnref{eqn:ChL}, we immediately obtain
\begin{multline}
\modS_{\overline{(\ell;r,0)} (\ell' , \lambda' ; \Delta_{r',s'})} = \sum_{s''=1}^{v-1} \brac{-1}^{s''-1} \sum_{\ell''=0}^{\infty} \left[ \modS_{(\ell + 2v \ell'' + s'' , \lambda_{r,s''} ; \Delta_{r,s''}) (\ell' , \lambda' ; \Delta_{r',s'})} \right. \\
\left. - \modS_{(\ell + 2v \brac{\ell'' + 1} - s'' , \lambda_{u-r,v-s''} ; \Delta_{u-r,v-s''}) (\ell' , \lambda' ; \Delta_{r',s'})} \right].
\end{multline}
Inserting the standard S-matrix elements \eqref{eqn:STypEntries} and simplifying using \eqref{eqn:DefLambda} and \eqref{eqn:KacSymmetry}, the right-hand side becomes
\begin{equation} \label{eqn:Careful}
\frac{1}{2} \frac{\abs{\tau}}{-\ii \tau} \ee^{-\ii \pi \brac{k \ell \ell' + \ell \lambda' + \ell' \brac{r-1}}} \sum_{\ell''=0}^{\infty} \ee^{-2 \pi \ii v \lambda' \ell''} \sum_{s''=1}^{v-1} \brac{-1}^{s''-1} \sqbrac{\ee^{-\ii \pi s'' \lambda'} - \ee^{\ii \pi \brac{s'' - 2v} \lambda'}} \vmodS_{(r,s'') (r',s')}.
\end{equation}
Performing the $\ell''$-sum and extracting the $s''$-dependent factors from $\vmodS_{(r,s'') (r',s')}$ now gives
\begin{equation}
\frac{1}{2} \frac{\abs{\tau}}{-\ii \tau} \ee^{-\ii \pi \brac{k \ell \ell' + \ell \lambda' + \ell' \brac{r-1}}} \sum_{s''=1}^{v-1} \brac{-1}^{\brac{r'-1} \brac{s''-1}} \frac{\sin \brac{\brac{v-s''} \pi \lambda'}}{\sin \brac{v \pi \lambda'}} \frac{\sin \brac{\pi s' s'' t}}{\sin \brac{\pi s' t}} \vmodS_{(r,1) (r',s')}.
\end{equation}
The result now follows from \lemref{lem:UsefulIdentity}.
\end{proof}

\noindent Of course, the T-transformations of the atypical characters are also easy to obtain.

As remarked above, this procedure would also allow us to determine the atypical S-matrix entries $\modS_{\overline{(\ell;r,s)} (\ell',\lambda' ; \Delta_{r',s'})}$, defined by transforming the characters of the $\sfmod{\ell}{\DiscMod{r,s}^+}$.  However, the resulting entries are not particularly pleasant to work with and we will see that we can proceed with our computations without their explicit form.

\begin{cor} \label{cor:VacMod}
Let $k$ be an admissible level with $v>1$.  Then, the S-matrix entries for the vacuum module $\IrrMod{0} = \IrrMod{1,0}$ take the form
\begin{equation} \label{eqn:SVacEntries}
\modS_{\overline{(0;1,0)} (\ell',\lambda' ; \Delta_{r',s'})} = \frac{1}{2} \frac{\abs{\tau}}{-\ii \tau} \frac{1}{2 \cos \brac{\pi \lambda'} + \brac{-1}^{r'} 2 \cos \brac{k \pi s'}} \vmodS_{(1,1) (r',s')}.
\end{equation}
\end{cor}

Finally, we make a few comments:  First, we remark that these atypical S-matrix ``entries'' are not really entries of the S-matrix because we have chosen our character basis to consist of those of the typical irreducibles $\sfmod{\ell}{\TypMod{\lambda ; \Delta_{r,s}}}$ and the atypical indecomposables $\sfmod{\ell}{\TypMod{\lambda_{r,s} ; \Delta_{r,s}}^+}$.  There is therefore no sense in trying to construct S-matrix entries involving two atypical irreducibles (these do not come from our basis).  Second, we note that the atypical S-matrix entry \eqref{eqn:SAtypEntries} diverges precisely when $\lambda'$ takes on the atypical values $\lambda' = \lambda_{r',s'}$ and $\lambda_{u-r',v-s'}$.  Our last comment is to admit that we have been a little cavalier with regard to the sum over $\ell''$ in \eqnref{eqn:Careful}.  We should be more careful here because the summand sits at the radius of convergence, hence a regularisation is in order \cite{CreW1p13}.  However, this will not affect the Verlinde computations of the next section.

\begin{ex}
When $v=2$, $s'$ is restricted to be $1$ and we have $\cos \brac{k \pi s'} = \cos \brac{u \pi / 2} = 0$.  The denominator of each atypical S-matrix entry therefore simplifies to $2 \cos \brac{\pi \lambda'}$.  This agrees with the result obtained for $k=-\tfrac{1}{2}$ in \cite{CreMod12}.  We can also recover the result reported there for $k=-\tfrac{4}{3}$ as then $r'$ is restricted to be $1$ and $s'$ to be $1$ or $2$.  For both choices of $s'$, the denominator simplifies to $2 \cos \brac{\pi \lambda'} + 1$.
\end{ex}

\section{The Verlinde Formula} \label{sec:Verlinde}

Having determined the S-matrix entries for both typical and atypical irreducibles, we can now consider the implications for the fusion rules of the admissible level theories.  For this, we use the continuum version of the Verlinde formula.  Of course, we expect indecomposable representations in the spectrum when $v>1$, so the Verlinde formula cannot tell us about the fusion ring directly, but rather it is expected to give the structure constants of the Grothendieck ring of fusion.  This is the quotient of the fusion ring by the ideal generated by the (formal) differences of each indecomposable and the direct sum of its composition factors.\footnote{That this is indeed an ideal, hence that the Grothendieck ring is well-defined, requires that fusion define an exact functor from the category of admissible modules to itself.  This is not guaranteed in general (see \cite{GabFus09} for examples), but we expect that fusion is exact for the fractional level theories studied here.}  In any case, we may use the Verlinde formula to define structure constants and investigate whether they seem to define reasonable Grothendieck fusion rings.  Our conjecture here is as follows:

\begin{conj} \label{conj:Verlinde}
Let $k$ be an admissible level with $v>1$.  Then, the continuum Verlinde formula
\begin{multline} \label{eqn:Verlinde}
\fuscoeff{(\ell , \lambda ; \Delta_{r,s}) (\ell' , \lambda' ; \Delta_{r',s'})}{(\ell'' , \lambda'' ; \Delta_{r'',s''})} \\
= \sum_{m \in \ZZ} \sideset{}{'} \sum_{R,S} \int_{-1}^1 \frac{\modS_{(\ell , \lambda ; \Delta_{r,s}) (m , \mu ; \Delta_{R,S})} \modS_{(\ell' , \lambda' ; \Delta_{r',s'}) (m , \mu ; \Delta_{R,S})} \modS_{(\ell'' , \lambda'' ; \Delta_{r'',s''}) (m , \mu ; \Delta_{R,S})}^*}{\modS_{\overline{(0;1,0)} (m , \mu ; \Delta_{R,S})}} \: \dd \mu,
\end{multline}
and its atypical generalisations with one or both of $(\ell , \lambda ; \Delta_{r,s})$ and $(\ell' , \lambda' ; \Delta_{r',s'})$ replaced by $\overline{(\ell ; r,s)}$ and $\overline{(\ell' ; r',s')}$, respectively, give the structure constants of the Grothendieck ring of fusion.  The Grothendieck fusion rules take the form
\begin{equation}
\Gr{\sfmod{\ell}{\TypMod{\lambda ; \Delta_{r,s}}}} \fuse \Gr{\sfmod{\ell'}{\TypMod{\lambda' ; \Delta_{r',s'}}}} = \sum_{\ell'' \in \ZZ} \ \sideset{}{'} \sum_{r'',s''} \int_{-1}^1 \fuscoeff{(\ell , \lambda ; \Delta_{r,s}) (\ell' , \lambda' ; \Delta_{r',s'})}{(\ell'' , \lambda'' ; \Delta_{r'',s''})} \Gr{\sfmod{\ell''}{\TypMod{\lambda'' ; \Delta_{r'',s''}}}} \: \dd \lambda'',
\end{equation}
along with their atypical generalisations.  Here, the square brackets $\tGr{\cdots}$ remind us that we are working with the Grothendieck quotient of the fusion ring.
\end{conj}

\noindent This conjecture has been verified for $k=-\tfrac{1}{2}$ in \cite{RidFus10}, subject only to the standard and well-tested assumption that fusion respects spectral flow.  It was also shown there that fusion is exact with the same assumption.  The conjecture has been similarly checked for the atypical $k=-\tfrac{4}{3}$ modules and one of the typical modules in \cite{CreMod12} (the fusion rules at this level were only computed in \cite{GabFus01} for one typical admissible).  We remark that computing the Grothendieck fusion rules using the Verlinde formula guarantees that they will be commutative and associative.

Let us turn to the computation of the Grothendieck fusion coefficients for the standard modules.  Substituting the S-matrix entries \eqref{eqn:STypEntries} and \eqref{eqn:SVacEntries}, the right-hand side of \eqref{eqn:Verlinde} becomes
\begin{multline}
\frac{1}{2} \sum_{m \in \ZZ} \ee^{-\ii \pi \brac{k \brac{\ell + \ell' - \ell''} + \lambda + \lambda' - \lambda''} m} \Biggl[ \int_{-1}^1 \ee^{-\ii \pi \brac{\ell + \ell' - \ell''} \mu} \cos \brac{\pi \mu} \: \dd \mu \sideset{}{'} \sum_{R,S} \frac{\vmodS_{(r,s) (R,S)} \vmodS_{(r',s') (R,S)} \vmodS_{(r'',s'') (R,S)}}{\vmodS_{(1,1) (R,S)}} \Biggr. \\
\Biggl. + \int_{-1}^1 \ee^{-\ii \pi \brac{\ell + \ell' - \ell''} \mu} \: \dd \mu \: \sideset{}{'} \sum_{R,S} \brac{-1}^R \cos \brac{k \pi S} \frac{\vmodS_{(r,s) (R,S)} \vmodS_{(r',s') (R,S)} \vmodS_{(r'',s'') (R,S)}}{\vmodS_{(1,1) (R,S)}} \Biggr].
\end{multline}
We recognise the Virasoro fusion coefficient $\vfuscoeff{(r,s) (r',s')}{(r'',s'')}$ as the sum over $R$ and $S$ in the first term.  The sum in the second term is similarly recognised after realising that
\begin{align}
\brac{-1}^R \cos \frac{u \pi S}{v} \vmodS_{(r',s') (R,S)} &= -\sqrt{\frac{2}{uv}} \brac{-1}^{r'S + R \brac{s' \pm 1}} \sin \frac{v \pi r' R}{u} \sqbrac{\sin \frac{u \pi \brac{s'-1} S}{v} + \sin \frac{u \pi \brac{s'+1} S}{v}} \notag \\
&= \frac{1}{2} \sqbrac{\vmodS_{(r',s'-1) (R,S)} + \vmodS_{(r',s'+1) (R,S)}}.
\end{align}
We remark that when $s'-1 = 0$ or $s'+1 = v$, the corresponding sine functions vanish in the above expression.  Thus, when the indices $s'-1$ and $s'+1$ fall out of the Kac table, the above Virasoro S-matrix entries should be understood to vanish.  With this proviso in mind, the sum over $m$ and integral over $\mu$ are now easily dealt with and we arrive at a general expression for the standard Grothendieck fusion coefficients:
\begin{multline} \label{eqn:FusCoeffTyp}
\fuscoeff{(\ell , \lambda ; \Delta_{r,s}) (\ell' , \lambda' ; \Delta_{r',s'})}{(\ell'' , \lambda'' ; \Delta_{r'',s''})} \\
= \Bigl[ \delta_{\ell'' = \ell + \ell' + 1} \func{\delta}{\lambda'' = \lambda + \lambda' - k \bmod{2}} + \delta_{\ell'' = \ell + \ell' - 1} \func{\delta}{\lambda'' = \lambda + \lambda' + k \bmod{2}} \Bigr] \vfuscoeff{(r,s) (r',s')}{(r'',s'')} \\
+ \delta_{\ell'' = \ell + \ell'} \func{\delta}{\lambda'' = \lambda + \lambda' \bmod{2}} \sqbrac{\vfuscoeff{(r,s) (r',s'-1)}{(r'',s'')} + \vfuscoeff{(r,s) (r',s'+1)}{(r'',s'')}}.
\end{multline}
Of course, the Virasoro fusion coefficients vanish too whenever an index falls outside the $\minmod{u}{v}$ Kac table.

\begin{prop} \label{prop:FusTyp}
Let $k$ be an admissible level with $v>1$.  Then, the Grothendieck fusion rules for the standard admissibles are given by
\begin{multline} \label{eqn:GrFusTyp}
\Gr{\sfmod{\ell}{\TypMod{\lambda ; \Delta_{r,s}}}} \fuse \Gr{\sfmod{\ell'}{\TypMod{\lambda' ; \Delta_{r',s'}}}} \\
= \sum_{r'',s''} \vfuscoeff{(r,s) (r',s')}{(r'',s'')} \brac{\Gr{\sfmod{\ell + \ell' + 1}{\TypMod{\lambda + \lambda' - k ; \Delta_{r'',s''}}}} + \Gr{\sfmod{\ell + \ell' - 1}{\TypMod{\lambda + \lambda' + k ; \Delta_{r'',s''}}}}} \\
+ \sum_{r'',s''} \brac{\vfuscoeff{(r,s) (r',s'-1)}{(r'',s'')} + \vfuscoeff{(r,s) (r',s'+1)}{(r'',s'')}} \Gr{\sfmod{\ell + \ell'}{\TypMod{\lambda + \lambda' ; \Delta_{r'',s''}}}}.
\end{multline}
\end{prop}

\noindent Of course, one can insert the known expressions for the Virasoro fusion coefficients into \eqref{eqn:FusCoeffTyp} and \eqref{eqn:GrFusTyp} to obtain completely explicit, if rather lengthy, formulae.  We recall that these coefficients have the form
\begin{subequations}
\begin{equation}
\vfuscoeff{(r,s) (r',s')}{(r'',s'')} = \vfuscoeff{(r,1) (r',1)}{(r'',1)} \vfuscoeff{(1,s) (1,s')}{(1,s'')} \equiv \vpfuscoeff{r, r'}{r''}{u} \vpfuscoeff{s, s'}{s''}{v},
\end{equation}
where
\begin{equation} \label{eqn:VirFusCoeff}
\vpfuscoeff{t, t'}{t''}{w} = 
\begin{cases}
1 & \text{if \(\abs{t-t'}+1 \leqslant t'' \leqslant \min \set{t+t'-1,2w-t-t'-1}\) and \(t+t'+t''\) is odd,} \\
0 & \text{otherwise.}
\end{cases}
\end{equation}
\end{subequations}
In particular, we note the following useful identities:
\begin{gather} \label{eqn:VirFusIdent}
\vpfuscoeff{1, t'}{t''}{w} = \delta_{t''=t'}, \qquad 
\vpfuscoeff{t, w-t'}{w-t''}{w} = \vpfuscoeff{t, t'}{t''}{w}, \qquad 
\vpfuscoeff{t, w-1}{t''}{w} = \delta_{t''=w-t}.
\end{gather}
It is a useful exercise to check that \eqref{eqn:GrFusTyp} is symmetric under $(r,s) \leftrightarrow (r',s')$, hence that this Grothendieck fusion rule is commutative.  We have also dropped the primes from the summations in this rule because $(r'',s'')$ and $(u-r'',v-s'')$ cannot appear together in $\minmod{u}{v}$ fusion rules.

Notice that the spectral flow indices in \eqref{eqn:FusCoeffTyp}, and hence in \eqref{eqn:GrFusTyp}, are always constrained so that the total spectral flow index is conserved, meaning that the right-hand sides depend upon the sum of the spectral flow indices of the modules being fused, rather than upon their individual indices.  This means that the Grothendieck fusion rules for the standard modules satisfy
\begin{equation} \label{eqn:GrFusionSF}
\Gr{\sfmod{\ell}{\mathcal{M}}} \fuse \Gr{\sfmod{\ell'}{\mathcal{M'}}} = \Gr{\sfmod{\ell + \ell'}{\mathcal{M} \fuse \mathcal{M}'}},
\end{equation}
We may therefore restrict to untwisted modules with $\ell = \ell' = 0$ without any loss of generality.  A well-known, but still open, conjecture asserts that spectral flow also respects the genuine fusion rules in the sense that the analogue of \eqref{eqn:GrFusionSF} holds.

The (Grothendieck) fusion with $\tGr{\TypMod{\lambda ; \Delta_{1,1}}}$ is particularly nice because it preserves the minimal model index $r$ and generates modules with general $s$ from those with $s=1$:
\begin{equation} \label{GrFR:TxT}
\Gr{\TypMod{\lambda ; \Delta_{1,1}}} \fuse \Gr{\TypMod{\mu ; \Delta_{r,s}}} = \Gr{\sfmod{}{\TypMod{\lambda + \mu - k ; \Delta_{r,s}}}} + \Gr{\sfmod{-1}{\TypMod{\lambda + \mu + k ; \Delta_{r,s}}}} + \Gr{\TypMod{\lambda + \mu ; \Delta_{r,s-1}}} + \Gr{\TypMod{\lambda + \mu ; \Delta_{r,s+1}}}.
\end{equation}
To generate the ``seed'' modules $\tGr{\TypMod{\lambda ; \Delta_{r,1}}}$, we fuse $\tGr{\TypMod{\lambda ; \Delta_{1,1}}}$ with the atypicals $\IrrMod{r,0} = \IrrMod{r-1}$.  In fact, the Grothendieck fusion of such an atypical  with a standard module is even easier to compute than that of two standard modules because the denominators appearing in \thmref{thm:AtypMod} and \corref{cor:VacMod} cancel.

\begin{prop} \label{prop:FusIrrTyp}
Let $k$ be an admissible level with $v>1$.  Then, the Grothendieck fusion of $\sfmod{\ell}{\IrrMod{r,0}} = \sfmod{\ell}{\IrrMod{r-1}}$ and $\sfmod{\ell'}{\TypMod{\lambda' ; \Delta_{r',s'}}}$ is given by
\begin{equation} \label{eqn:FusIrrTyp}
\Gr{\sfmod{\ell}{\IrrMod{r,0}}} \fuse \Gr{\sfmod{\ell'}{\TypMod{\lambda' ; \Delta_{r',s'}}}} = \sum_{r''} \vfuscoeff{(r,1) (r',1)}{(r'',1)} \Gr{\sfmod{\ell + \ell'}{\TypMod{r-1 + \lambda' ; \Delta_{r'',s'}}}}.
\end{equation}
\end{prop}

\noindent This confirms \eqref{eqn:GrFusionSF} once again.  If $u>2$, we may take $r=2$ and deduce that
\begin{equation} \label{GrFR:LxT}
\Gr{\IrrMod{1}} \fuse \Gr{\TypMod{\mu ; \Delta_{r,s}}} = \Gr{\TypMod{\mu - 1 ; \Delta_{r-1,s}}} + \Gr{\TypMod{\mu + 1 ; \Delta_{r+1,s}}}.
\end{equation}
Thus, one can generate the seeds $\tGr{\TypMod{\lambda ; \Delta_{r,1}}}$, with $r>1$, by fusing $\tGr{\TypMod{\lambda ; \Delta_{1,1}}}$ repeatedly with $\tGr{\IrrMod{1}}$.  Fusing these seeds repeatedly with $\tGr{\TypMod{\lambda ; \Delta_{1,1}}}$ then generates the remaining standard modules $\tGr{\TypMod{\lambda ; \Delta_{r,s}}}$.

We remark that if $\TypMod{\mu ; \Delta_{r,s}}$ is irreducible, meaning that $\mu \neq \lambda_{r,s} , \lambda_{u-r,v-s} \bmod{2}$, then so are $\TypMod{\mu - 1 ; \Delta_{r-1,s}}$ and $\TypMod{\mu + 1 ; \Delta_{r+1,s}}$.  Moreover, the conformal dimensions of the \hwss{} of these irreducibles satisfy
\begin{equation} \label{eqn:DimsDon'tMatch}
\Delta_{r+1,s} - \Delta_{r-1,s} = \frac{r}{t} - s = \frac{r}{u}v - s \notin \ZZ.
\end{equation}
We may therefore conclude that the modules $\TypMod{\mu - 1 ; \Delta_{r-1,s}}$ and $\TypMod{\mu + 1 ; \Delta_{r+1,s}}$ appearing on the right-hand side of \eqref{GrFR:LxT} may not be combined into a single indecomposable module.  \eqnDref{GrFR:LxT}{eqn:DimsDon'tMatch} then imply the \emph{genuine} fusion rule
\begin{equation} \label{FR:LxT}
\IrrMod{1} \fuse \TypMod{\mu ; \Delta_{r,s}} = \TypMod{\mu - 1 ; \Delta_{r-1,s}} \oplus \TypMod{\mu + 1 ; \Delta_{r+1,s}} \qquad \text{(\(\mu \neq \lambda_{r,s} , \lambda_{u-r,v-s} \bmod{2}\)).}
\end{equation}
Of course, this deduction is contingent upon the validity of our continuum Verlinde formula.  One can similarly deduce that the typical fusion rule \eqref{GrFR:TxT} implies the corresponding genuine fusion rule for generic $\lambda$ and $\mu$, more precisely for
\begin{equation}
\lambda \neq \pm k \bmod{2}, \qquad 
\mu \neq \lambda_{r,s}, \lambda_{u-r,v-s} \bmod{2}, \qquad 
\lambda + \mu \neq 0, 1, \lambda_{r,s} \pm k, \lambda_{u-r,v-s} \pm k.
\end{equation}

We now turn to the Grothendieck fusion of the atypicals $\IrrMod{r,0}$ with one another.

\begin{prop} \label{prop:FusIrrIrr}
Let $k$ be an admissible level with $u>2$ and $v>1$.  Then, the Grothendieck fusion of the $\sfmod{\ell}{\IrrMod{r,0}}$ is given by
\begin{equation} \label{GrFR:LxL}
\Gr{\sfmod{\ell}{\IrrMod{r,0}}} \fuse \Gr{\sfmod{\ell'}{\IrrMod{r',0}}} = \sum_{r''} \vfuscoeff{(r,1) (r',1)}{(r'',1)} \Gr{\sfmod{\ell + \ell'}{\IrrMod{r'',0}}}.
\end{equation}
\end{prop}
\begin{proof}
This time, the denominators of the atypical S-matrix entries do not cancel and we would have to apply \lemref{lem:UsefulIdentity} to the integrand of the Verlinde formula in order to proceed.  However, it turns out to be much easier to combine \eqref{eqn:ChL} with \propref{prop:FusIrrTyp} in this case.  First, note that because \eqref{eqn:GrFusionSF} holds for Verlinde computations, we may assume that $\ell = \ell' = 0$ for simplicity.  Now,
\begin{align}
\Gr{\IrrMod{r,0}} \fuse \Gr{\IrrMod{r',0}} &= \sum_{s=1}^{v-1} \brac{-1}^{s-1} \sum_{\ell = 0}^{\infty} \set{\Gr{\sfmod{2v \ell + s}{\TypMod{r,s}}} - \Gr{\sfmod{2v \brac{\ell + 1} - s}{\TypMod{u-r,v-s}}}} \fuse \Gr{\IrrMod{r',0}} \notag \\
&= \sum_{s=1}^{v-1} \brac{-1}^{s-1} \sum_{\ell = 0}^{\infty} \sum_{r'',s''} \vfuscoeff{(r,s) (r',1)}{(r'',s'')} \notag \\
&\mspace{50mu} \cdot \set{\Gr{\sfmod{2v \ell + s}{\TypMod{\lambda_{r,s} + r'-1; \Delta_{r'',s''}}}} - \Gr{\sfmod{2v \brac{\ell + 1} - s}{\TypMod{\lambda_{u-r,v-s} + r'-1; \Delta_{r'',s''}}}}},
\end{align}
where we have used Kac symmetry to identify $\vfuscoeff{(r,s) (r',1)}{(r'',s'')}$ with $\vfuscoeff{(u-r,v-s) (r',1)}{(r'',s'')}$.  We now note that $\vfuscoeff{(r,s) (r',1)}{(r'',s'')} = \vfuscoeff{(r,1) (r',1)}{(r'',1)} \delta_{s''=s}$, hence that we may replace $\lambda_{r,s} + r'-1 = \lambda_{r+r'-1,s}$ by $\lambda_{r'',s''}$ (because $r+r'-1 = r'' \bmod 2$ when the fusion coefficient is non-zero) and, similarly, $\lambda_{u-r,v-s} + r'-1$ by $\lambda_{u-r'',v-s''}$.  Using \eqref{eqn:ChL} once again, we obtain
\begin{equation}
\Gr{\IrrMod{r,0}} \fuse \Gr{\IrrMod{r',0}} = \sum_{r''} \vfuscoeff{(r,1) (r',1)}{(r'',1)} \Gr{\IrrMod{r'',0}},
\end{equation}
as required.
\end{proof}

\noindent Putting $r'=2$, we obtain
\begin{equation} \label{eqn:GrFR:LxL1}
\Gr{\IrrMod{r}} \fuse \Gr{\IrrMod{1}} = \Gr{\IrrMod{r-1}} + \Gr{\IrrMod{r+1}},
\end{equation}
with $\IrrMod{-1} = \IrrMod{u-1} = \set{0}$ as usual.  It is easy to check that this result always lifts to the genuine fusion ring, hence that \eqref{GrFR:LxL} does too (using associativity).  Our (conjectured) Verlinde formula therefore implies the following result:

\begin{thm} \label{thm:FusionSubring}
Let $k$ be an admissible level and let $\fusring{k}$ denote the fusion ring generated by the admissible level $k$ $\AKMA{sl}{2}$-modules.  Then, the subring of $\fusring{k}$ generated by the $\IrrMod{r}$, with $r=0,1,\ldots,u-2$, is isomorphic to the non-negative integer level fusion ring $\fusring{u-2}$.  In particular, the irreducible module $\IrrMod{u-2}$ is a simple current in $\fusring{k}$ of dimension $\frac{1}{4} \brac{u-2}v$.
\end{thm}

\noindent We remark that this theorem holds for $v=1$, where the simple current $\IrrMod{u-2} = \IrrMod{k}$ has dimension $\frac{1}{4} k$.

\begin{ex}
For $k=-\tfrac{1}{2}$, $u=3$ and the theorem says that $\IrrMod{0}$ and $\IrrMod{1}$ generate a subring isomorphic to the fusion ring of $\AKMA{sl}{2}$ at level $1$.  Thus, we indeed have a simple current:  $\IrrMod{1} \fuse \IrrMod{1} = \IrrMod{0}$.  For $k=-\tfrac{4}{3}$, $u=2$, so the theorem only tells us that the vacuum module $\IrrMod{0}$ generates a subring isomorphic to the fusion ring of the trivial theory --- the simple current guaranteed by the theorem is only non-trivial when $u>2$.
\end{ex}

Finally, the Grothendieck fusion rules of the $\DiscMod{r,s}^+$ and their images under spectral flow follow, with a little effort, from \eqref{eqn:Induct} and the rules already determined.  We will need the following identities pertaining to Virasoro fusion coefficients in addition to those stated in \eqref{eqn:VirFusIdent}.  They follow directly from the explicit formula \eqref{eqn:VirFusCoeff}.

\begin{lem}
The factorised Virasoro fusion coefficients $\vpfuscoeff{t,t'}{t''}{w}$ satisfy the following identities:
\begin{subequations}
\begin{align}
\vpfuscoeff{t,t'}{t''}{w} - \vpfuscoeff{t+1,t'+1}{t''}{w} &= \delta_{t''=2w-t-t'-1} & &\text{(\(t+t' > w-1\)),} \label{eqn:FusCoeff1} \\
\vpfuscoeff{t,t'}{t''}{w} - \vpfuscoeff{t+1,t'+1}{t''}{w} &= 0 & &\text{(\(t+t' = w-1\)),} \label{eqn:FusCoeff2} \\
\vpfuscoeff{t+1,t'+1}{t''}{w} - \vpfuscoeff{t,t'}{t''}{w} &= \delta_{t''=t+t'+1} & &\text{(\(t+t' < w-1\)).} \label{eqn:FusCoeff3}
\end{align}
\end{subequations}
\end{lem}

\begin{prop} \label{prop:FusDrs}
Let $k$ be an admissible level with $v>1$.  Then, the Grothendieck fusion rules involving the $\sfmod{\ell}{\DiscMod{r,s}^+}$ are
\begin{align}
\Gr{\sfmod{\ell}{\IrrMod{r,0}}} \fuse \Gr{\sfmod{\ell'}{\DiscMod{r',s'}^+}} &= \sum_{r''} \vfuscoeff{(r,1) (r',1)}{(r'',1)} \Gr{\sfmod{\ell + \ell'}{\DiscMod{r'',s'}^+}}, \label{GrFR:LxD} \\
\Gr{\sfmod{\ell}{\TypMod{\lambda ; \Delta_{r,s}}}} \fuse \Gr{\sfmod{\ell'}{\DiscMod{r',s'}^+}} &= \sum_{r'',s''} \vfuscoeff{(r,s) (r',s'+1)}{(r'',s'')} \Gr{\sfmod{\ell + \ell'}{\TypMod{\lambda + \lambda_{r',s'} ; \Delta_{r'',s''}}}} \notag \\
&\mspace{55mu} + \sum_{r'',s''} \vfuscoeff{(r,s) (r',s')}{(r'',s'')} \Gr{\sfmod{\ell + \ell' + 1}{\TypMod{\lambda + \lambda_{r',s'+1} ; \Delta_{r'',s''}}}}, \label{GrFR:ExD} \\
\Gr{\sfmod{\ell}{\DiscMod{r,s}^+}} \fuse \Gr{\sfmod{\ell'}{\DiscMod{r',s'}^+}} &= 
\begin{cases}
\displaystyle \sum_{r'',s''} \vfuscoeff{(r,s) (r',s')}{(r'',s'')} \Gr{\sfmod{\ell + \ell' + 1}{\TypMod{\lambda_{r'',s+s'+1}; \Delta_{r'',s''}}}} \\
\displaystyle \mspace{20mu} + \sum_{r''} \vfuscoeff{(r,1) (r',1)}{(r'',1)} \Gr{\sfmod{\ell + \ell'}{\DiscMod{r'',s+s'}^+}}, & \text{if \(s+s'<v\),} \\
\displaystyle \sum_{r'',s''} \vfuscoeff{(r,s+1) (r',s'+1)}{(r'',s'')} \Gr{\sfmod{\ell + \ell' + 1}{\TypMod{\lambda_{r'',s+s'+1}; \Delta_{r'',s''}}}} \\
\displaystyle \mspace{20mu} + \sum_{r''} \vfuscoeff{(r,1) (r',1)}{(r'',1)} \Gr{\sfmod{\ell + \ell' + 1}{\DiscMod{u-r'',s+s'-v+1}^+}}, & \text{if \(s+s' \geqslant v\).}
\end{cases}
\label{GrFR:DxD}
\end{align}
\end{prop}
\begin{proof}
As the Verlinde formula respects spectral flow, we may simplify our calculations by assuming that $\ell = \ell' = 0$.  Probably the easiest proofs of these Grothendieck fusion rules are by induction on $s'$ using \eqref{es:Drs}.  We first detail the argument for \eqref{GrFR:LxD}.  The base case is $s' = v-1$ for which $\DiscMod{r',s'}^+ = \sfmod{}{\IrrMod{u-r',0}}$:
\begin{align}
\Gr{\IrrMod{r,0}} \fuse \Gr{\DiscMod{r',v-1}^+} &= \Gr{\IrrMod{r,0}} \fuse \Gr{\sfmod{}{\IrrMod{u-r',0}}} = \sum_{r''} \vfuscoeff{(r,1) (u-r',1)}{(r'',1)} \Gr{\sfmod{}{\IrrMod{r'',0}}} \notag \\
&= \sum_{r''} \vfuscoeff{(r,1) (u-r',1)}{(u-r'',1)} \Gr{\sfmod{}{\IrrMod{u-r'',0}}} = \sum_{r''} \vfuscoeff{(r,1) (r',1)}{(r'',1)} \Gr{\sfmod{}{\IrrMod{u-r'',0}}} \notag \\
&= \sum_{r''} \vfuscoeff{(r,1) (r',1)}{(r'',1)} \Gr{\DiscMod{r'',v-1}^+}.
\end{align}
Here, we have used \propref{prop:FusIrrIrr}, shifted $r''$ to $u-r''$, and employed the second identity of \eqref{eqn:VirFusIdent}.  Assuming that \eqref{GrFR:LxD} holds for a given $s'+1$, we obtain
\begin{align}
\Gr{\IrrMod{r,0}} \fuse \Gr{\DiscMod{r',s'}^+} &= \Gr{\IrrMod{r,0}} \fuse \brac{\Gr{\sfmod{}{\TypMod{\lambda_{r',s'+1}; \Delta_{r',s'+1}}}} - \Gr{\sfmod{}{\DiscMod{r',s'+1}^+}}} \notag \\
&= \sum_{r'',s''} \vfuscoeff{(r,1) (r',s'+1)}{(r'',s'')} \Gr{\sfmod{}{\TypMod{\lambda_{r+r'-1,s'+1}; \Delta_{r'',s''}}}} - \sum_{r''} \vfuscoeff{(r,1) (r',1)}{(r'',1)} \Gr{\sfmod{}{\DiscMod{r'',s'+1}^+}} \notag \\
&= \sum_{r''} \vfuscoeff{(r,1) (r',1)}{(r'',1)} \brac{\Gr{\sfmod{}{\TypMod{\lambda_{r'',s'+1}; \Delta_{r'',s'+1}}}} - \Gr{\sfmod{}{\DiscMod{r'',s'+1}^+}}} \notag \\
&= \sum_{r''} \vfuscoeff{(r,1) (r',1)}{(r'',1)} \Gr{\sfmod{\ell + \ell'}{\DiscMod{r'',s'}^+}},
\end{align}
using \eqref{eqn:Induct}, \propref{prop:FusIrrTyp}, the first identity of \eqref{eqn:VirFusIdent}, and noting that the Virasoro fusion coefficient vanishes unless $r'' = r+r'-1 \bmod{2}$.

\eqnref{GrFR:ExD} follows by a similar argument involving \propref{prop:FusTyp}, though this time we need all the identities of \eqref{eqn:VirFusIdent}.  We omit the details.  \eqnref{GrFR:DxD} is the most involved.  The base case $s' = v-1$ proceeds smoothly using \eqref{GrFR:LxD} and noting that $s+s'$ is necessarily $v$ or greater.  We remark that in this case, the first sum on the right-hand side of \eqref{GrFR:DxD} vanishes because $(r',s'+1) = (r',v)$ falls outside the Kac table.

To tackle the induction step with $s'<v-1$, we use \eqref{GrFR:ExD} and \eqref{eqn:Induct} to derive that
\begin{multline} \label{eqn:Derive}
\Gr{\DiscMod{r,s}^+} \fuse \Gr{\DiscMod{r',s'}^+} = \sum_{r''} \vpfuscoeff{r,r'}{r''}{u} \bigg\{ \sum_{s''} \vpfuscoeff{s+1,s'+1}{s''}{v} \Gr{\sfmod{}{\TypMod{\lambda_{r'',s+s'+1}; \Delta_{r'',s''}}}} \\
+ \sum_{s''} \vpfuscoeff{s,s'+1}{s''}{v} \Gr{\sfmod{2}{\TypMod{\lambda_{r'',s+s'+2}; \Delta_{r'',s''}}}} \bigg\} - \Gr{\DiscMod{r,s}^+} \fuse \Gr{\sfmod{}{\DiscMod{r',s'+1}^+}},
\end{multline}
into which we substitute the appropriate version of \eqref{GrFR:DxD}.  There are three cases to consider:
\begin{enumerate}[leftmargin=*]
\item When $s+s' \geqslant v$, the first term of the substitution almost precisely cancels the second term of \eqref{eqn:Derive}.  Indeed, the identity \eqref{eqn:FusCoeff1} lets us replace the second line of \eqref{eqn:Derive} by
\begin{equation} \label{eqn:Recognise}
\Gr{\sfmod{2}{\TypMod{\lambda_{r'',s+s'+2}; \Delta_{r'',2v-s-s'-2}}}} - \Gr{\sfmod{2}{\DiscMod{u-r'',s+s'-v+2}^+}}
\end{equation}
which, after a little massaging, we recognise from \eqref{eqn:Induct} as $\tGr{\sfmod{}{\DiscMod{u-r'',s+s'-v+1}^+}}$.
\item When $s+s'=v-1$, everything proceeds as in the previous case except that \eqref{eqn:Recognise} is now noted to be $\tGr{\sfmod{}{\IrrMod{u-r'',0}}} = \tGr{\DiscMod{r'',v-1}^+}$ by \eqref{es:Lr}.  The identity \eqref{eqn:FusCoeff2} takes care of the first line of \eqref{eqn:Derive}.
\item Finally, if $s+s' \leqslant v-2$, then the first term of the substitution perfectly cancels the second term of \eqref{eqn:Derive}.  We therefore have to use \eqref{eqn:FusCoeff3} to isolate the term $\tGr{\sfmod{}{\TypMod{\lambda_{r'',s+s'+1}; \Delta_{r'',s+s'+1}}}}$ from the first term of \eqref{eqn:Derive} (thereby leaving the first term in the required form).  Up to this first term, the right-hand side of \eqref{eqn:Derive} then becomes $\tGr{\sfmod{}{\TypMod{\lambda_{r'',s+s'+1}; \Delta_{r'',s+s'+1}}}} - \tGr{\sfmod{}{\DiscMod{r'',s+s'+1}^+}} = \tGr{\DiscMod{r'',s+s'}^+}$. \qedhere
\end{enumerate}
\end{proof}

\noindent It is straight-forward, though a little tedious, to analyse when these Grothendieck fusion rules lift to genuine fusion rules.  For example, \eqref{GrFR:LxD} always does.

Explicit formulae aside, an important consequence of these computations is the following:

\begin{thm} \label{thm:GrFusCoeffPos}
The Grothendieck fusion coefficients are non-negative integers.
\end{thm}

\noindent This is an extremely important consistency check for our conjectured continuum Verlinde formula, the positivity strongly supporting the truth of the conjecture.  Granting this, we feel justified in claiming that the longstanding problem of obtaining a sensible Verlinde formula for fractional level \WZW{} models has been solved, at least for $\AKMA{sl}{2}$.

\section{Comparison with Koh and Sorba} \label{sec:KohSorba}

Of course, the initial difficulty encountered when trying to apply the standard Verlinde formula \cite{KacMod88} for highest weight admissible modules to fractional level models was that the resulting ``fusion coefficients'' were often negative integers.  This was first noted by Koh and Sorba in \cite{KohFus88}.  These negative Verlinde coefficients were shown in \cite{RidSL208}, for $k=-\tfrac{1}{2}$, to be consequences of treating the characters of the highest weight admissible modules as theta functions and implicitly continuing them outside their correct convergence regions \eqref{eqn:Convergence}.  Specifically, it was noted that the correct $k=-\tfrac{1}{2}$ fusion rules reduce to those deduced by Koh and Sorba from the standard Verlinde formula if one imposes the character identity \eqref{eqn:ZeroCharacter} (and its spectral flow versions) at the level of modules.  This reduction was also explicitly checked to reproduce all negative coefficients for $k=-\tfrac{4}{3}$ in \cite{CreMod12}.

In the general formalism we have developed here, this reduction procedure amounts to setting all the standard modules to zero in the Grothendieck fusion ring.  The following corollary of the computations of \secref{sec:Verlinde} is therefore pertinent:

\begin{prop} \label{prop:GrFRIdeal}
For $v>1$, the standard modules generate an ideal of the Grothendieck fusion ring.
\end{prop}

\noindent We claim that the structure coefficients of the quotient of the Grothendieck fusion ring by the ideal of (equivalence classes of) standard modules are precisely the Verlinde coefficients computed by Koh and Sorba.  Settling this claim will then confirm that the explanation detailed in \cite{RidSL208} for the negative ``fusion coefficients'' is correct for all admissible $\AKMA{sl}{2}$ theories.  We have no doubt that this continues to hold for admissible theories based on higher rank semisimple Lie algebras.

Before demonstrating our claim, we pause to note that we do not claim that the standard modules generate an ideal of the genuine fusion ring itself.  The known fusion rules for $k=-\tfrac{1}{2}$ and $k=-\tfrac{4}{3}$ \cite{GabFus01,RidFus10} show that this assertion is false.  Rather, we expect that the irreducible standard modules (the typical modules) are \emph{projective} in an appropriate category of $\AKMA{sl}{2}_k$-modules (vertex algebra modules) and that they span, together with the projective covers of the atypical irreducibles, an ideal of the fusion ring.  We hope to return to questions of projectivity in the future.

To compare with Koh and Sorba, we first present a short dictionary to translate their notation:
\begin{center}
\begin{tabular}{c||C|C|C|C|C|C|C}
KS & m & t & u & n & k & \phi_{0:n} & \phi_{k:n} \\
\hline
CR & k & u-2v & v & r-1 & s & \IrrMod{n+1,0} & \DiscMod{n+1,k}^+
\end{tabular}
\end{center}
Modulo an obvious typo, their fusion rules \cite[Eq.~(14)]{KohFus88} become
\begin{equation} \label{eqn:KohSorba}
\KSGr{\DiscMod{r,s}^+} \fuse \KSGr{\DiscMod{r',s'}^+} = 
\begin{cases}
\displaystyle +\sum_{r''} \vfuscoeff{(r,1) (r',1)}{(r'',1)} \KSGr{\DiscMod{r'',s+s'}^+} & \text{if \(s+s'<v\),} \\
\displaystyle -\sum_{r''} \vfuscoeff{(r,1) (r',1)}{(r'',1)} \KSGr{\DiscMod{u-r'',s+s'-v}^+} & \text{if \(s+s' \geqslant v\).}
\end{cases}
\end{equation}
where we let $\DiscMod{r,0}^+ \equiv \IrrMod{r,0}$ for convenience and use angled brackets $\tKSGr{\cdots}$ in anticipation of quotienting the Grothendieck fusion ring (where elements are indicated with square brackets $\tGr{\cdots}$) by its standard ideal.

\begin{thm} \label{thm:KSDerived}
The Grothendieck fusion rules reduce to the fusion rules of Koh and Sorba upon quotienting by the ideal of standard modules.
\end{thm}
\begin{proof}
This is a straight-forward check.  Consider the Grothendieck fusion rule \eqref{GrFR:DxD} with $\ell = \ell' = 0$.  Setting all (equivalence classes of) standard modules to zero, this becomes
\begin{equation}
\KSGr{\DiscMod{r,s}^+} \fuse \KSGr{\DiscMod{r',s'}^+} = \sum_{r''} \vfuscoeff{(r,1) (r',1)}{(r'',1)} \KSGr{\DiscMod{r'',s+s'}^+}
\end{equation}
for $s+s'<v$, in agreement with \eqref{eqn:KohSorba}.  When $s+s' \geqslant v$, we have to use in addition \eqref{eqn:ZeroCharacter} and then \eqref{eqn:SFRelations}:
\begin{align}
\KSGr{\DiscMod{r,s}^+} \fuse \KSGr{\DiscMod{r',s'}^+} &= +\sum_{r''} \vfuscoeff{(r,1) (r',1)}{(r'',1)} \KSGr{\sfmod{}{\DiscMod{u-r'',s+s'-v+1}^+}} \notag \\
&= -\sum_{r''} \vfuscoeff{(r,1) (r',1)}{(r'',1)} \KSGr{\sfmod{}{\DiscMod{r'',2v-s-s'-1}^-}} \notag \\
&= -\sum_{r''} \vfuscoeff{(r,1) (r',1)}{(r'',1)} \KSGr{\DiscMod{u-r'',s+s'-v}^+}.
\end{align}
Similar considerations for \eqref{GrFR:LxL} and \eqref{GrFR:LxD} now complete the proof.
\end{proof}

\noindent Finally, we cannot resist recording the following amusing summary of this result:\footnote{We blame Simon Wood for this quip and direct any complaints towards his general direction.}  In order to recover the well-known negative fusion coefficients for fractional level $\AKMA{sl}{2}$ \WZW{} models, we have to set our standards to zero.

\section{Examples} \label{sec:Examples}

We now illustrate the results of our Verlinde formula computations by specialising to certain admissible levels.  These levels will be chosen to have the added benefit that the extended algebras defined by the simple currents $\IrrMod{u-2}$ are (potentially) interesting.  Determining the algebraic structure of simple current extensions is straight-forward, though there are certain subtleties that arise.  We refer to \cite{RidSU206} (see also \cite{RidSL208}) for a detailed account of these, quoting only the results (adapted to the $\SLA{sl}{2;\RR}$ adjoint) that we require for our examples.

Let the zero-grade subspace of the simple current $\IrrMod{u-2}$ be spanned by states $\ket{\psi^{(n)}}$ of weight $u-2 - 2n$, for $n = 0, 1, \ldots, u-2$.  These states all have conformal dimension $\Delta = \tfrac{1}{4} \brac{u-2} v$.  Define constants $\eps_n$, $\mu_{m,n}$ and $\mu_{J,n}$ (for appropriate $J \in \SLA{sl}{2}$) by
\begin{equation}
\tbrac{\psi^{(m)}_r}^{\dag} = \eps_m \psi^{(\lambda - m)}_{-r}, \qquad 
\begin{aligned}
\func{J}{z} \func{\psi^{(n)}}{w} &= \mu_{J,n} \func{\psi^{(n)}}{w} \func{J}{z}, \\
\func{\psi^{(m)}}{z} \func{\psi^{(n)}}{w} &= \mu_{m,n} \func{\psi^{(n)}}{w} \func{\psi^{(m)}}{z},
\end{aligned}
\end{equation}
where $\func{\psi^{(n)}}{w}$ is the field corresponding to $\ket{\psi^{(n)}}$ and the $\psi^{(n)}_r$ are its modes.  We obtain, as in \cite{RidSU206}, the following results:
\begin{itemize}[leftmargin=*]
\item The $\mu_{J,n}$ are real and independent of $n$.  Moreover, $\mu_{h,n} = 1$, whereas $\mu_{e,n} \mu_{f,n} = 1$.
\item The $\eps_m$ are all related by $\eps_m = \mu_{f,n}^m \eps_0$.
\item For any $u$, the choice $\mu_{e,n} = \mu_{f,n} = \eps_m = 1$ is consistent.\footnote{In fact, there is a second consistent choice when $u$ is even:  $\mu_{e,n} = \mu_{f,n} = -1$, $\eps_m = \brac{-1}^m$ and $\mu_{m,n} = \brac{-1}^{2 \Delta + m+n}$.  The fact that there exist different consistent choices (for $u$ even) for the extension field localities reflects the choice that we have in extending the $\SLA{sl}{2;\RR}$ adjoint to the extended algebra.}
\item Assuming that $2 \Delta \in \ZZ$ ($uv$ is even), this choice leads to $\mu_{m,n} = \brac{-1}^{2 \Delta + u} = \brac{-1}^{u+v+uv/2}$.
\end{itemize}
Note that the constants $\mu_{J,n}$ and $\mu_{m,n}$ are mutual locality indices so, for example, the extension fields $\func{\psi^{(n)}}{w}$ are all mutually bosonic with respect to $\func{h}{z}$.

\subsection*{Example:  $k=-\tfrac{1}{2}$}

We start with the familiar case of $t=\frac{3}{2}$, giving $u=3$, $v=2$ and $c=-1$.  From the table in \figref{fig:Tables}, we see that the standard modules all have the form $\sfmod{\ell}{\TypMod{\lambda;-1/8}}$, with $\lambda \in \RR / 2 \ZZ$, and the irreducible atypicals have the form $\sfmod{\ell}{\IrrMod{\mu}}$, with $\mu \in \set{0,1}$.  The Grothendieck fusion rules imply the genuine fusion rules
\begin{equation}
\IrrMod{1} \fuse \IrrMod{1} = \IrrMod{0}, \qquad 
\IrrMod{1} \fuse \TypMod{\lambda;-1/8} = \TypMod{\lambda + 1;-1/8} \qquad \text{(\(\lambda \neq \pm \tfrac{1}{2} \bmod{2}\))}
\end{equation}
and spectral flow extends this by \eqref{eqn:GrFusionSF}.\footnote{In this section, we will present all fusion rules in the untwisted sector for clarity.  We will also omit explicitly noting the fusion rules involving the vacuum module which acts as the fusion identity.}  The standard module Grothendieck fusion rules similarly imply that
\begin{equation}
\TypMod{\lambda;-1/8} \fuse \TypMod{\mu;-1/8} = \sfmod{}{\TypMod{\lambda + \mu + 1/2;-1/8}} \oplus \sfmod{-1}{\TypMod{\lambda + \mu - 1/2;-1/8}},
\end{equation}
but only when $\lambda + \mu \neq 0,1 \bmod{2}$ --- otherwise, comparing the conformal dimensions of states in the fusion product does not guarantee complete reducibility.  Indeed, the fusion of standard modules with $\lambda + \mu \in \ZZ$ was shown in \cite{RidFus10} to yield indecomposable modules $\StagMod{\lambda + \mu}$ on which the Virasoro mode $L_0$ acts non-diagonalisably (staggered modules).

Since $u$ is odd and the simple current $\IrrMod{1}$ has zero grade fields of dimension $\Delta = \frac{1}{2}$, we can conclude that the extension fields are mutually bosonic with respect to the affine fields and one another:  $\brac{-1}^{2 \Delta + u} = 1$.  Computing the extended algebra is very simple because the non-regular \opes{} are summarised by
\begin{equation}
\func{\psi^{(0)}}{z} \func{\psi^{(1)}}{w} \sim \frac{-1}{z-w},
\end{equation}
which we identify as describing the $\beta \gamma$ ghost algebra.  We remark that it was shown in \cite{RidSL208} that this simple current extension fails to be associative if we had chosen the $\SLA{su}{2}$ adjoint.

\subsection*{Example:  $k=-\tfrac{4}{3}$}

For this level, we have $u=2$, $v=3$ and $c=-6$, so the standard modules all have the form $\sfmod{\ell}{\TypMod{\lambda;-1/3}}$, with $\lambda \in \RR / 2 \ZZ$.  This time (see \figref{fig:Tables}), the atypical irreducibles fall into two classes $\sfmod{\ell}{\IrrMod{0}}$ and $\sfmod{\ell}{\DiscMod{-2/3}^+}$.  The above results for the Grothendieck fusion now imply the following fusion rule for the standard modules:
\begin{equation}
\TypMod{\lambda;-1/3} \fuse \TypMod{\mu;-1/3} = \sfmod{}{\TypMod{\lambda + \mu + 4/3;-1/3}} \oplus \TypMod{\lambda + \mu;-1/3} \oplus \sfmod{-1}{\TypMod{\lambda + \mu - 4/3;-1/3}}.
\end{equation}
This time, we can only be sure that the fusion product is completely reducible for $\lambda + \mu \neq 0,1, \pm \tfrac{2}{3} \bmod{2}$.  It was shown in \cite{GabFus01} that taking $\lambda = \mu = 0$ gives a staggered module.  We expect that $\lambda + \mu = 0, \pm \tfrac{2}{3}$ always gives staggered modules whereas $\lambda + \mu = 1$ does not.

The Grothendieck fusion rules involving $\DiscMod{-2/3}^+ = \DiscMod{1,1}^+$ now follow from \propref{prop:FusDrs}:
\begin{subequations}
\begin{align}
\TypMod{\lambda;-1/3} \fuse \DiscMod{-2/3}^+ &= \TypMod{\lambda - 2/3;-1/3} \oplus \sfmod{}{\TypMod{\lambda + 2/3;-1/3}} \qquad \text{(\(\lambda \neq 0 \bmod{2}\)),} \\
\DiscMod{-2/3}^+ \fuse \DiscMod{-2/3}^+ &= \sfmod{}{\TypMod{\lambda_{1,3};-1/3}} \oplus \DiscMod{-4/3}^+ = \sfmod{}{\TypMod{0;-1/3}} \oplus \sfmod{}{\IrrMod{0}}.
\end{align}
\end{subequations}
The first fusion rule was shown to be staggered in \cite{GabFus01} when $\lambda = 0 \bmod{2}$.

\subsection*{Example:  $k=+\tfrac{1}{2}$}

For this admissible level, we have $c=\tfrac{3}{5}$, $u=5$ and $v=2$, so there are four atypical irreducibles and two continuous families of standard modules (up to spectral flow):
\[
\set{\IrrMod{0}, \IrrMod{1}, \IrrMod{2}, \IrrMod{3}; \TypMod{\lambda; 1/8}, \TypMod{\lambda; -3/40}}.
\]
We refer to \figref{fig:Tables} for the admissible weights $\lambda_{r,s}$ and conformal dimensions $\Delta_{r,s}$ that characterise $k=\tfrac{1}{2}$.  In particular, we note that the standard modules are typical (irreducible) for all $\lambda \neq \pm \tfrac{1}{2} \bmod{2}$.

We summarise the Grothendieck fusion rules of the admissible irreducibles in \tabref{tab:FRk=1/2}.  Note that the fusion rules of the $\IrrMod{\lambda}$ are those of $\AKMA{sl}{2}_3$, in agreement with \thmref{thm:FusionSubring}.  These rules all lift to genuine fusion rules (Grothendieck sums are replaced by direct sums), except when we are fusing typicals with typicals and the weight labels sum to $\lambda + \mu \in \ZZ$.  For example,
\begin{equation}
\TypMod{\lambda; 1/8} \fuse \TypMod{\mu; -3/40} = \sfmod{}{\TypMod{\lambda + \mu - 1/2; -3/40}} \oplus \sfmod{-1}{\TypMod{\lambda + \mu + 1/2; -3/40}} \qquad \text{(\(\lambda + \mu \notin \ZZ\)).}
\end{equation}

{
\begin{table}
\begin{center}
\setlength{\extrarowheight}{6pt}
\begin{tabular}{C|CCCCC}
\fuse & (1) & (2) & (3) & [\tfrac{1}{8}]_{\mu} & [-\tfrac{3}{40}]_{\mu} \\[1mm]
\hline
(1) & (0)+(2) & (1)+(3) & (2) & [-\tfrac{3}{40}]_{\mu+1} & [\tfrac{1}{8}]_{\mu+1} + [-\tfrac{3}{40}]_{\mu+1} \\[1mm]
(2) & \star & (0)+(2) & (1) & [-\tfrac{3}{40}]_{\mu} & [\tfrac{1}{8}]_{\mu} + [-\tfrac{3}{40}]_{\mu} \\[1mm]
(3) & \star & \star & (0) & [\tfrac{1}{8}]_{\mu+1} & [-\tfrac{3}{40}]_{\mu+1} \\[1mm]
[\tfrac{1}{8}]_{\lambda} & \star & \star & \star & [\tfrac{1}{8}]_{\lambda+\mu-1/2}^{1} + [\tfrac{1}{8}]_{\lambda+\mu+1/2}^{-1} & [-\tfrac{3}{40}]_{\lambda+\mu-1/2}^{1} + [-\tfrac{3}{40}]_{\lambda+\mu+1/2}^{-1} \\[1mm]
[-\tfrac{3}{40}]_{\lambda} & \star & \star & \star & \star & \begin{matrix} [\tfrac{1}{8}]_{\lambda+\mu-1/2}^{1} + [\tfrac{1}{8}]_{\lambda+\mu+1/2}^{-1} \\ + [-\tfrac{3}{40}]_{\lambda+\mu-1/2}^{1} + [-\tfrac{3}{40}]_{\lambda+\mu+1/2}^{-1} \end{matrix}
\end{tabular}
\vspace{3mm}
\caption{The (Grothendieck) fusion rules of the admissible modules, up to spectral flow, when $k=+\tfrac{1}{2}$.  The notation $(\lambda)$ stands for $\IrrMod{\lambda}$ and $[\Delta]_{\lambda}^{\ell}$ for $\sfmod{\ell}{\TypMod{\lambda; \Delta}}$.  The stars are entries that we omit for clarity (fusion is commutative).} \label{tab:FRk=1/2}
\end{center}
\end{table}
}

The cautious reader will have noticed that the conformal dimensions of the states comprising the fusion product of $\TypMod{\lambda; -3/40}$ and $\TypMod{\mu; -3/40}$ do not forbid the possibility of reducible but indecomposable modules when $\lambda + \mu = \pm \tfrac{1}{5} \mod{2}$.  In this case, \tabref{tab:FRk=1/2} gives four Grothendieck summands and it seems that either the first and fourth or the second and third summands might combine into a single indecomposable.  However, this is ruled out by associativity:
\begin{align}
\TypMod{\lambda; -3/40} \fuse \TypMod{\mu; -3/40} &= \IrrMod{1} \fuse \TypMod{\lambda - 1; 1/8} \fuse \TypMod{\mu; -3/40} = \IrrMod{1} \fuse \Bigl( \sfmod{}{\TypMod{\lambda + \mu + 1/2; -3/40}} \oplus \sfmod{-1}{\TypMod{\lambda + \mu - 1/2; -3/40}} \Bigr) \notag \\
&= \sfmod{}{\TypMod{\lambda + \mu - 1/2; 1/8}} \oplus \sfmod{-1}{\TypMod{\lambda + \mu + 1/2; 1/8}} \notag \\
&\phantom{= \sfmod{}{\TypMod{\lambda + \mu - 1/2; 1/8}}} \mspace{4.4mu} \oplus \sfmod{}{\TypMod{\lambda + \mu - 1/2; -3/40}} \oplus \sfmod{-1}{\TypMod{\lambda + \mu + 1/2; -3/40}},
\end{align}
as $\lambda + \mu = \pm \tfrac{1}{5} \mod{2}$ clearly implies that $\lambda + \mu \notin \ZZ$.  By contrast, we believe that all of the fusion rules involving two standard modules yield staggered modules when $\lambda + \mu \in \ZZ$.

The simple current $\IrrMod{3}$ has $\Delta = \tfrac{3}{2}$.  The extended algebra is therefore generated by the three $\AKMA{sl}{2}$ currents and four \emph{bosonic} dimension $\tfrac{3}{2}$ fields.  We will not list the \opes{} of the latter because we have not managed to identify the extended algebra conclusively.  However, we believe that it coincides with the quantum hamiltonian reduction of $\affine{\alg{g}}_2$ at level $-\tfrac{3}{2}$, where the $\SLA{sl}{2}$ embedding lies along the highest root of $\alg{g}_2$.  This is supported by the following facts:  The central charge of this reduction is indeed $\tfrac{3}{5}$ (see \cite[Eq.~(4.4)]{KacQua03} for example); this reduction is strongly generated by bosonic fields, three of dimension $1$, four of dimension $\tfrac{3}{2}$ and one of dimension $2$ (the energy-momentum field); and the dimension $1$ fields of this reduction generate a copy of $\AKMA{sl}{2}$ at level $k=\tfrac{1}{2}$ (see \cite[Eq.~(4.6) and Prop.~4.1]{KacQua03}).

\subsection*{Example:  $k=-\tfrac{2}{3}$}

This time, $u=4$, $v=3$ and $c=-\tfrac{3}{2}$.  Up to spectral flow, there are six families of atypical irreducibles and three families of standard modules represented by
\[
\set{\IrrMod{0}, \IrrMod{1}, \IrrMod{2}; \DiscMod{-4/3}^+, \DiscMod{-1/3}^+, \DiscMod{2/3}^+; \TypMod{\lambda; -1/6}, \TypMod{\lambda; -5/48}, \TypMod{\lambda; 1/3}}.
\]
Once again, the admissible highest weights and conformal dimensions may be found in \figref{fig:Tables}.  We summarise the Grothendieck fusion rules in \tabref{tab:FRk=-2/3}.

As with $k=\tfrac{1}{2}$, the simple current $\IrrMod{2}$ has $\Delta = \tfrac{3}{2}$, but this time the fields are fermionic.  We therefore expect that the simple current extension will be the $N=3$ superconformal algebra of central charge $c=-\tfrac{3}{2}$.  More precisely, we expect to obtain the reduced $N=3$ superconformal algebra that results from decoupling the fermionic dimension $\tfrac{1}{2}$ field \cite{GodFac88} (which obviously does not appear in our extension).  This algebra is generated by three dimension $1$ fields $\func{J^a}{z}$ and three dimension $\tfrac{3}{2}$ fields $\func{G^a}{z}$ satisfying
\begin{equation} \label{OPE:N=3}
\begin{gathered}
\func{J^a}{z} \func{J^b}{w} \sim \frac{\kappa^{ab} k}{\brac{z-w}^2} + \frac{{\mathsf{f}^{ab}}_c \func{J^c}{w}}{z-w}, \qquad 
\func{J^a}{z} \func{G^b}{w} \sim \frac{{\mathsf{f}^{ab}}_c \func{G^c}{w}}{z-w}, \\
\begin{aligned}
\func{G^a}{z} \func{G^b}{w} &\sim \frac{2 \kappa^{ab} \brac{k-1}}{\brac{z-w}^3} + \frac{2 {\mathsf{f}^{ab}}_c \brac{k-1} \func{J^c}{w} / k}{\brac{z-w}^2} \\
&\phantom{\sim} + \frac{4 \kappa^{ab} \func{T}{w} + {\mathsf{f}^{ab}}_c \partial \func{J^c}{w} - 2 \normord{\func{J^a}{w} \func{J^b}{w}} / k}{z-w},
\end{aligned}
\end{gathered}
\end{equation}
where $\kappa^{ab}$ and ${\mathsf{f}^{ab}}_c$ represent the trace form and structure constants of $\SLA{sl}{2}$, respectively.  The ``level'' $k$ parametrises the central charge of the reduced $N=3$ superconformal algebra by $c = \tfrac{1}{2} \brac{3k-1}$.  Computing the extended algebra \opes{} precisely reproduces \eqref{OPE:N=3} if we identify the $J^a$ with the $\AKMA{sl}{2}$-currents, the $G^a$ with the simple current fields and set $k=-\tfrac{2}{3}$.  We remark that the modules labelled by $r=1$ or $3$ combine, under the extended algebra action, to give Neveu-Schwarz sector modules, whereas those labelled by $r=2$ yield Ramond sector modules.

\subsection*{Example:  $k=-\tfrac{5}{4}$}

For our last example, $u=3$, $v=4$ and $c=-5$.  The admissible highest weights and conformal dimensions are likewise given in \figref{fig:Tables}.  We see again that there are six families of atypical irreducibles and three families of standard modules, up to spectral flow:
\[
\set{\IrrMod{0}, \IrrMod{1}; \DiscMod{-3/4}^+, \DiscMod{1/4}^+, \DiscMod{-3/2}^+, \DiscMod{-1/2}^+; \TypMod{\lambda; -5/16}, \TypMod{\lambda; -1/4}, \TypMod{\lambda; 3/16}}.
\]
The Grothendieck fusion rules for this model are collected in \tabref{tab:FRk=-5/4}.  We remark that these rules show that fusion multiplicities for admissible level theories can be greater than $2$ (see that of $\TypMod{\lambda + \mu; -1/4}$ in $\TypMod{\lambda; -1/4} \fuse \TypMod{\mu; -1/4}$).

For this level, the simple current $\IrrMod{1}$ has $\Delta = 1$, so one expects that the extended algebra will be of affine type.  There are two extension fields and they are fermionic, hence the extended algebra is bound to be the affine Kac-Moody superalgebra $\AKMSA{osp}{1}{2}_{-5/4}$.  The $\AKMA{sl}{2}_{-5/4}$ fields generate the bosonic subalgebra and the extension fields should provide the remaining fermionic generators.  Noting that the central charges of $\AKMSA{osp}{1}{2}_{-5/4}$ and $\AKMA{sl}{2}_{-5/4}$ are indeed equal ($c=-5$), this identification amounts to a conformal embedding.

To verify this, we identify $\SLA{sl}{2}$ with the bosonic subalgebra of $\SLSA{osp}{1}{2}$ and let the fermionic basis elements be given, in the defining representation, by
\begin{equation}
\psi^+ = 
\begin{pmatrix}
0 & 0 & 1 \\
0 & 0 & -\ii \\
-1 & \ii & 0
\end{pmatrix}
, \qquad \psi^- = 
\begin{pmatrix}
0 & 0 & 1 \\
0 & 0 & \ii \\
1 & \ii & 0
\end{pmatrix}
.
\end{equation}
The non-vanishing (anti)commutation relations are then those of $\SLA{sl}{2}$, augmented by
\begin{equation} \label{eqn:osp12}
\begin{aligned}
\comm{h}{\psi^{\pm}} &= \pm \psi^{\pm}, \\
\acomm{\psi^+}{\psi^+} &= 4e,
\end{aligned}
\qquad
\begin{aligned}
\comm{e}{\psi^-} &= -\psi^+, \\
\acomm{\psi^+}{\psi^-} &= 2h,
\end{aligned}
\qquad
\begin{aligned}
\comm{f}{\psi^+} &= \psi^-, \\
\acomm{\psi^-}{\psi^-} &= 4f
\end{aligned}
\end{equation}
and the supertrace form is that of $\SLA{sl}{2}$, augmented by
\begin{equation} \label{eqn:osp12'}
\killing{\psi^+}{\psi^-} = -\killing{\psi^-}{\psi^+} = 4.
\end{equation}
Denoting the two dimension $1$ fields of $\IrrMod{1}$ by $\func{\psi^+}{z}$ and $\func{\psi^-}{z}$ and normalising them appropriately, we compute their operator product expansions using the methods of \cite{RidSU206}:
\begin{equation}
\begin{aligned}
\func{h}{z} \func{\psi^{\pm}}{w} &\sim \pm \frac{\func{\psi^{\pm}}{w}}{z-w}, \\
\func{\psi^+}{z} \func{\psi^+}{w} &\sim \frac{4 \func{e}{w}}{z-w},
\end{aligned}
\quad
\begin{aligned}
\func{e}{z} \func{\psi^-}{w} &\sim -\frac{\func{\psi^+}{w}}{z-w}, \\
\func{\psi^+}{z} \func{\psi^-}{w} &\sim -\frac{5}{\brac{z-w}^2} + \frac{2 \func{h}{w}}{z-w},
\end{aligned}
\quad
\begin{aligned}
\func{f}{z} \func{\psi^+}{w} &\sim \frac{\func{\psi^-}{w}}{z-w}, \\
\func{\psi^-}{z} \func{\psi^-}{w} &\sim \frac{4 \func{f}{w}}{z-w}.
\end{aligned}
\end{equation}
Comparing with \eqref{eqn:osp12} and \eqref{eqn:osp12'} demonstrates that this simple current extension of $\AKMA{sl}{2}_{-5/4}$ is indeed $\AKMSA{osp}{1}{2}_{-5/4}$.

We remark that it is very easy to check explicitly that $k=-\tfrac{5}{4}$ is an admissible level for $\AKMSA{osp}{1}{2}$ and that the relation
\begin{equation}
\brac{\psi^+_{-2} - 4 \psi^-_{-1} e_{-1} + 2 h_{-1} \psi^-_{-1}} \ket{0} = 0
\end{equation}
holds in the (irreducible) $\AKMSA{osp}{1}{2}$ vacuum module.  This constrains the spectrum so that \hwss{} must have weight ($h_0$-eigenvalue) $0$ or $-\tfrac{1}{2}$ and the remaining relaxed \hwss{} must have conformal dimension $-\tfrac{1}{4}$.  The corresponding relaxed highest weight $\AKMSA{osp}{1}{2}_{-5/4}$-modules are clearly formed by combining the $\AKMA{sl}{2}_{-5/4}$-modules $\IrrMod{0}$ with $\IrrMod{1}$, $\DiscMod{-1/2}^+$ with $\DiscMod{-3/2}^+$, and $\TypMod{\lambda; -1/4}$ with $\TypMod{\lambda + 1; -1/4}$.  Similarly combining modules whose label $s$ is odd leads to \emph{twisted} $\AKMSA{osp}{1}{2}_{-5/4}$-modules (the fermions act with half-integer moding).  It would be interesting to investigate the role played, if any, by such twisted superalgebra modules in physical applications.

{
\begin{sidewaystable}
\centering
\setlength{\extrarowheight}{7pt}
\vspace{0.6\textwidth}
\begin{tabular}{C|CCCCCCCC}
\fuse & (1) & (2) & (-\tfrac{4}{3}) & (-\tfrac{1}{3}) & (\tfrac{2}{3}) & [-\tfrac{1}{6}]_{\mu} & [-\tfrac{5}{48}]_{\mu} & [\tfrac{1}{3}]_{\mu} \\[2mm]
\hline
(1) & (0) + (2) & (1) & (-\tfrac{1}{3}) & (-\tfrac{4}{3}) + (\tfrac{2}{3}) & (-\tfrac{1}{3}) & [-\tfrac{5}{48}]_{\mu+1} & [-\tfrac{1}{6}]_{\mu+1} + [\tfrac{1}{3}]_{\mu+1} & [-\tfrac{5}{48}]_{\mu+1} \\[2mm]
(2) & \star & (0) & (\tfrac{2}{3}) & (-\tfrac{1}{3}) & (-\tfrac{4}{3}) & [\tfrac{1}{3}]_{\mu} & [-\tfrac{5}{48}]_{\mu} & [-\tfrac{1}{6}]_{\mu} \\[2mm]
(-\tfrac{4}{3}) & \star & \star & [-\tfrac{1}{6}]_{0}^{1} + (2)^1 & [-\tfrac{5}{48}]_{1}^{1} + (1)^1 & [\tfrac{1}{3}]_{0}^{1} + (0)^1 & [\tfrac{1}{3}]_{\mu+2/3} + [-\tfrac{1}{6}]_{\mu-2/3}^{1} & [-\tfrac{5}{48}]_{\mu+2/3} + [-\tfrac{5}{48}]_{\mu-2/3}^{1} & [-\tfrac{1}{6}]_{\mu+2/3} + [\tfrac{1}{3}]_{\mu-2/3}^{1} \\[2mm]
(-\tfrac{1}{3}) & \star & \star & \star & \begin{matrix} [-\tfrac{1}{6}]_{0}^{1} + (2)^1 \\ + [\tfrac{1}{3}]_{0}^{1} + (0)^1 \end{matrix} & [-\tfrac{5}{48}]_{1}^{1} + (1)^1 & [-\tfrac{5}{48}]_{\mu-1/3} + [-\tfrac{5}{48}]_{\mu+1/3}^{1} & \begin{matrix} [-\tfrac{1}{6}]_{\mu-1/3} + [\tfrac{1}{3}]_{\mu+1/3}^{1} \\ + [-\tfrac{1}{6}]_{\mu-1/3} + [\tfrac{1}{3}]_{\mu+1/3}^{1} \end{matrix} & [-\tfrac{5}{48}]_{\mu-1/3} + [-\tfrac{5}{48}]_{\mu+1/3}^{1} \\[2mm]
(\tfrac{2}{3}) & \star & \star & \star & \star & [-\tfrac{1}{6}]_{0}^{1} + (2)^1 & [\tfrac{1}{3}]_{\mu+2/3} + [-\tfrac{1}{6}]_{\mu-2/3}^{1} & [-\tfrac{5}{48}]_{\mu+2/3} + [-\tfrac{5}{48}]_{\mu-2/3}^{1} & [-\tfrac{1}{6}]_{\mu+2/3} + [\tfrac{1}{3}]_{\mu-2/3}^{1} \\[2mm]
[-\tfrac{1}{6}]_{\lambda} & \star & \star & \star & \star & \star & \begin{matrix} [-\tfrac{1}{6}]_{\lambda+\mu+2/3}^{1} + [-\tfrac{1}{6}]_{\lambda+\mu-2/3}^{-1} \\ + [\tfrac{1}{3}]_{\lambda+\mu} \end{matrix} & \begin{matrix} [-\tfrac{5}{48}]_{\lambda+\mu+2/3}^{1} + [-\tfrac{5}{48}]_{\lambda+\mu-2/3}^{-1} \\ + [-\tfrac{5}{48}]_{\lambda+\mu} \end{matrix} & \begin{matrix} [\tfrac{1}{3}]_{\lambda+\mu+2/3}^{1} + [\tfrac{1}{3}]_{\lambda+\mu-2/3}^{-1} \\ + [-\tfrac{1}{6}]_{\lambda+\mu} \end{matrix} \\[2mm]
[-\tfrac{5}{48}]_{\lambda} & \star & \star & \star & \star & \star & \star & \begin{matrix} [-\tfrac{1}{6}]_{\lambda+\mu+2/3}^{1} + [-\tfrac{1}{6}]_{\lambda+\mu-2/3}^{-1} \\ + [-\tfrac{1}{6}]_{\lambda+\mu} + [\tfrac{1}{3}]_{\lambda+\mu} \\ + [\tfrac{1}{3}]_{\lambda+\mu+2/3}^{1} + [\tfrac{1}{3}]_{\lambda+\mu-2/3}^{-1} \end{matrix} & \begin{matrix} [-\tfrac{5}{48}]_{\lambda+\mu+2/3}^{1} + [-\tfrac{5}{48}]_{\lambda+\mu-2/3}^{-1} \\ + [-\tfrac{5}{48}]_{\lambda+\mu} \end{matrix} \\[2mm]
[\tfrac{1}{3}]_{\lambda} & \star & \star & \star & \star & \star & \star & \star & \begin{matrix} [-\tfrac{1}{6}]_{\lambda+\mu+2/3}^{1} + [-\tfrac{1}{6}]_{\lambda+\mu-2/3}^{-1} \\ + [\tfrac{1}{3}]_{\lambda+\mu} \end{matrix}
\end{tabular}
\vspace{5mm}
\caption{The (Grothendieck) fusion rules of the admissible modules, up to spectral flow, when $k=-\tfrac{2}{3}$.  The notation $(\lambda)$ stands for $\IrrMod{\lambda}$ or $\DiscMod{\lambda}^+$ for $\lambda \in \NN$ or not, respectively, and $[\Delta]_{\lambda}$ for $\TypMod{\lambda; \Delta}$.  A superscript $\ell$ indicates that $\sfaut^{\ell}$ has been applied.} \label{tab:FRk=-2/3}
\end{sidewaystable}
}

{
\begin{sidewaystable}
\centering
\setlength{\extrarowheight}{7pt}
\vspace{0.6\textwidth}
\begin{tabular}{C|CCCCCCCC}
\fuse & (1) & (-\tfrac{3}{4}) & (\tfrac{1}{4}) & (-\tfrac{3}{2}) & (-\tfrac{1}{2}) & [-\tfrac{5}{16}]_{\mu} & [-\tfrac{1}{4}]_{\mu} & [\tfrac{3}{16}]_{\mu} \\[2mm]
\hline
(1) & (0) & (\tfrac{1}{4}) & (-\tfrac{3}{4}) & (-\tfrac{1}{2}) & (-\tfrac{3}{2}) & [\tfrac{3}{16}]_{\mu+1} & [-\tfrac{1}{4}]_{\mu+1} & [-\tfrac{5}{16}]_{\mu+1} \\[2mm]
(-\tfrac{3}{4}) & \star & \begin{matrix} [-\tfrac{5}{16}]_{-1/4}^{1} \\ + (-\tfrac{3}{2}) \end{matrix} & \begin{matrix} [\tfrac{3}{16}]_{3/4}^{1} \\ + (-\tfrac{1}{2}) \end{matrix} & \begin{matrix} [-\tfrac{1}{4}]_{1}^{1} \\ + (1)^{1} \end{matrix} & \begin{matrix} [-\tfrac{1}{4}]_{0}^{1} \\ + (0)^{1} \end{matrix} & [-\tfrac{1}{4}]_{\mu-3/4} + [-\tfrac{5}{16}]_{\mu+1/2}^{1} & \begin{matrix} [-\tfrac{5}{16}]_{\mu-3/4} + [\tfrac{3}{16}]_{\mu-3/4} \\ + [-\tfrac{1}{4}]_{\mu+1/2}^{1} \end{matrix} & [-\tfrac{1}{4}]_{\mu-3/4} + [\tfrac{3}{16}]_{\mu+1/2}^{1} \\[2mm]
(\tfrac{1}{4}) & \star & \star & \begin{matrix} [-\tfrac{5}{16}]_{-1/4}^{1} \\ + (-\tfrac{3}{2}) \end{matrix} & \begin{matrix} [-\tfrac{1}{4}]_{0}^{1} \\ + (0)^{1} \end{matrix} & \begin{matrix} [-\tfrac{1}{4}]_{1}^{1} \\ + (1)^{1} \end{matrix} & [-\tfrac{1}{4}]_{\mu+1/4} + [\tfrac{3}{16}]_{\mu-1/2}^{1} & \begin{matrix} [-\tfrac{5}{16}]_{\mu+1/4} + [\tfrac{3}{16}]_{\mu+1/4} \\ + [-\tfrac{1}{4}]_{\mu-1/2}^{1} \end{matrix} & [-\tfrac{1}{4}]_{\mu+1/4} + [-\tfrac{5}{16}]_{\mu-1/2}^{1} \\[2mm]
(-\tfrac{3}{2}) & \star & \star & \star & \begin{matrix} [-\tfrac{5}{16}]_{1/4}^{1} \\ + (\tfrac{1}{4})^{1} \end{matrix} & \begin{matrix} [\tfrac{3}{16}]_{-3/4}^{1} \\ + (-\tfrac{3}{4})^{1} \end{matrix} & [\tfrac{3}{16}]_{\mu+1/2} + [-\tfrac{1}{4}]_{\mu-1/4}^{1} & \begin{matrix} [-\tfrac{1}{4}]_{\mu+1/2} \\ + [-\tfrac{5}{16}]_{\mu-1/4}^{1} + [\tfrac{3}{16}]_{\mu-1/4}^{1} \end{matrix} & [-\tfrac{5}{16}]_{\mu+1/2} + [-\tfrac{1}{4}]_{\mu-1/4}^{1} \\[2mm]
(-\tfrac{1}{2}) & \star & \star & \star & \star & \begin{matrix} [-\tfrac{5}{16}]_{1/4}^{1} \\ + (\tfrac{1}{4})^{1} \end{matrix}  & [-\tfrac{5}{16}]_{\mu-1/2} + [-\tfrac{1}{4}]_{\mu+3/4}^{1} & \begin{matrix} [-\tfrac{1}{4}]_{\mu-1/2} \\ + [-\tfrac{5}{16}]_{\mu+3/4}^{1} + [\tfrac{3}{16}]_{\mu+3/4}^{1} \end{matrix} & [\tfrac{3}{16}]_{\mu-1/2} + [-\tfrac{1}{4}]_{\mu+3/4}^{1}\\[2mm]
[-\tfrac{5}{16}]_{\lambda} & \star & \star & \star & \star & \star & \begin{matrix} [-\tfrac{5}{16}]_{\lambda+\mu-3/4}^{1} + [-\tfrac{5}{16}]_{\lambda+\mu+3/4}^{-1} \\ + [-\tfrac{1}{4}]_{\lambda+\mu} \end{matrix} & \begin{matrix} [-\tfrac{1}{4}]_{\lambda+\mu-3/4}^{1} + [-\tfrac{1}{4}]_{\lambda+\mu+3/4}^{-1} \\ + [-\tfrac{5}{16}]_{\lambda+\mu} + [\tfrac{3}{16}]_{\lambda+\mu} \end{matrix} & \begin{matrix} [\tfrac{3}{16}]_{\lambda+\mu-3/4}^{1} + [\tfrac{3}{16}]_{\lambda+\mu+3/4}^{-1} \\ + [-\tfrac{1}{4}]_{\lambda+\mu} \end{matrix}\\[2mm]
[-\tfrac{1}{4}]_{\lambda} & \star & \star & \star & \star & \star & \star & \begin{matrix} [-\tfrac{5}{16}]_{\lambda+\mu-3/4}^{1} + [-\tfrac{5}{16}]_{\lambda+\mu+3/4}^{-1} \\ + [\tfrac{3}{16}]_{\lambda+\mu-3/4}^{1} + [\tfrac{3}{16}]_{\lambda+\mu+3/4}^{-1} \\ + 2 \: [-\tfrac{1}{4}]_{\lambda+\mu} \end{matrix} & \begin{matrix} [-\tfrac{1}{4}]_{\lambda+\mu-3/4}^{1} + [-\tfrac{1}{4}]_{\lambda+\mu+3/4}^{-1} \\ + [-\tfrac{5}{16}]_{\lambda+\mu} + [\tfrac{3}{16}]_{\lambda+\mu} \end{matrix} \\[2mm]
[\tfrac{3}{16}]_{\lambda} & \star & \star & \star & \star & \star & \star & \star & \begin{matrix} [-\tfrac{5}{16}]_{\lambda+\mu-3/4}^{1} + [-\tfrac{5}{16}]_{\lambda+\mu+3/4}^{-1} \\ + [-\tfrac{1}{4}]_{\lambda+\mu} \end{matrix}
\end{tabular}
\vspace{5mm}
\caption{The (Grothendieck) fusion rules of the admissible modules, up to spectral flow, when $k=-\tfrac{5}{4}$.  The notation follows that of \tabref{tab:FRk=-2/3}.} \label{tab:FRk=-5/4}
\end{sidewaystable}
}

\section*{Acknowledgements}

We would like to thank Tomoyuki Arakawa, Pierre Mathieu, Yvan Saint-Aubin, Akihiro Tsuchiya and Simon Wood for valuable discussions relating to the results reported here.  The research of DR is supported by an Australian Research Council Discovery Project DP1093910.

\newpage

\raggedright

%\bibliography{mod2}
%\bibliographystyle{unsrt}

\end{document}